\documentclass[lettersize,journal]{IEEEtran}
\usepackage{amsmath}
\usepackage{amsfonts}
\usepackage{amssymb}
\usepackage{amsthm}
\usepackage{tabularx}
\usepackage{mathrsfs}
\usepackage{algorithmic}
\usepackage{algorithm}

\usepackage{url}
\usepackage[colorlinks]{hyperref}
\newtheorem{theorem}{Theorem}
\newtheorem{remark}{Remark}
\newtheorem{definition}{Definition}

\newtheorem{corollary}{Corollary}

\usepackage{textcomp}
\usepackage{stfloats}
\usepackage{float}
\usepackage{graphicx}
\usepackage{verbatim}
\usepackage{xcolor}
\usepackage{subfigure}

\renewcommand{\vec}[1]{\mathbf{#1}}

\makeatletter
\def\blfootnote{\xdef\@thefnmark{}\@footnotetext}
\makeatother

\begin{document}
	% under Arbitrarily Correlated Fading
\title{A Gaussian Copula Approach to the Performance Analysis of Fluid Antenna Systems\thanks{The work of F. Rostami Ghadi, K. Wong, and K. Tong is supported by the Engineering and Physical Sciences Research Council (EPSRC) under Grant EP/W026813/1. For the purpose of open access, the authors will apply a Creative Commons Attribution (CC BY) licence to any Author Accepted Manuscript version arising. The work of L\'opez-Mart\'inez was funded in part by Consejer\'ia de Transformaci\'on Econ\'omica, Industria, Conocimiento y Universidades of Junta de Andaluc\'ia, and in part by MICIU/AEI/10.13039/501100011033 through grants EMERGIA20\_00297 and PID2020-118139RB-I00. The work of C.-B. Chae is supported by the Institute of Information and Communication Technology Promotion (IITP) grant funded by the Ministry of Science and ICT (MSIT), Korea (No. 2024-00428780, No. 2021-0-00486).}
\thanks{F. Javier L\'opez-Mart\'inez is with the Communications and Signal Processing Lab, Telecommunication Research Institute (TELMA), Universidad de M\'alaga, M\'alaga, 29010, Spain. He is also with the Department of Signal Theory, Networking and Communications, University of Granada, 18071, Granada, Spain (e-mail: $\rm fjlm@ugr.es$).}
\thanks{F. R. Ghadi, K. K. Wong, and K. F. Tong are with the Department of Electronic and Electrical Engineering, University College London, London WC1E 6BT, United Kingdom. K. K. Wong is also with Yonsei Frontier Lab, Yonsei University, Seoul, 03722, Korea (e-mail: $\{\rm f.rostamighadi,\rm kai\text{-}kit.wong,k.tong\}@ucl.ac.uk$).}
\thanks{C. B. Chae is with School of Integrated Technology, Yonsei University, Seoul, 03722, Korea (e-mail: $\rm cbchae@yonsei.ac.kr$).}
\thanks{Y. Zhang is with Kuang-Chi Science Limited, Hong Kong SAR, China (e-mail: $\rm yangyang.zhang@kuang\text{-}chi.com$).}
\thanks{Corresponding author: Kai-Kit Wong.}
} 

\author{Farshad~Rostami~Ghadi, \IEEEmembership{Member}, \textit{IEEE},  
            Kai-Kit~Wong, \IEEEmembership{Fellow}, \textit{IEEE},
  	    F. Javier~L\'opez-Mart\'inez, \IEEEmembership{Senior Member}, \textit{IEEE}, 
	    Chan-Byoung Chae, \IEEEmembership{Fellow}, \textit{IEEE},
	    Kin-Fai Tong, \IEEEmembership{Fellow}, \textit{IEEE}, and Yangyang Zhang

\vspace{-7mm}
}
\maketitle

\begin{abstract}
This paper investigates the performance of a single-user fluid antenna system (FAS), by exploiting a class of elliptical copulas to describe the dependence structure amongst the fluid antenna positions (ports). By expressing the well-known Jakes' model in terms of the Gaussian copula, we consider two cases: (i) the general case, i.e., any arbitrary correlated fading distribution; and (ii) the specific case, i.e., correlated Nakagami-$m$ fading. For both scenarios, we first derive analytical expressions for the cumulative distribution function (CDF) and probability density function (PDF) of the equivalent channel in terms of multivariate normal distribution. Then we obtain the outage probability (OP) and the delay outage rate (DOR) to analyze the performance of FAS. By employing the popular rank correlation coefficients such as Spearman's~$\rho$ and Kendall's~$\tau$, we measure the degree of dependency in correlated arbitrary fading channels and illustrate how the Gaussian copula can be accurately connected to Jakes' model in FAS. Our numerical results demonstrate that increasing the size of FAS provides lower OP and DOR, but the system performance saturates as the number of antenna ports increases. In addition, our results indicate that FAS provides better performance compared to conventional single-fixed antenna systems even when the size of fluid antenna is small. 
\end{abstract}

%\vspace{-3.5ex}

\maketitle

\blfootnote{Digital Object Identifier 10.1109/XXX.2021.XXXXXXX}
%\IEEEpeerreviewmaketitle

\section{Introduction}\label{sec-intro}
\IEEEPARstart{R}{ecent} advances in diversity and spatial multiplexing techniques have led to massive multiple-input multiple-output (MIMO) being a key technology for the fifth-generation (5G) wireless communication systems, where a large number of antennas are equipped in the form of an antenna array at a base station \cite{von2018advanced,larsson2014massive,zhang2019cell,marzetta2015massive}. Even though the large number of antennas at the base stations in massive MIMO systems improves multiplexing gains and network capacity, the same increment in the number of antennas at user equipment (UE) is not anticipated. Scaling up the number of antennas at the UE in massive MIMO systems brings many challenges in terms of power consumption, complexity and cost of manufacturing, signal processing requirements, channel estimation, spatial separation, and etc. Even in high-frequency bands where the antenna size is smaller, the spatial separation between antennas in a small space that is customarily required can be a problem when the number of antennas increases, let alone the cost of increasing the number of radio frequency (RF) chains.

To tackle this issue, fluid antenna systems (FAS) have been recently introduced as an emerging technology that promises to achieve a remarkable diversity gain in the small space of mobile devices in sixth-generation (6G) wireless networks \cite{wong2020fluid,wong2023fluid,New-2023,New-twc2023}. In fact, FAS refers to a system where the antenna has the ability to switch its position (i.e., ports) instantly\footnote{Port switching generally leads to a switching delay; however, such a delay can be reasonably ignored for reconfigurable pixel-based fluid antennas \cite{song2013efficient}. On the other hand, if fluid antennas are implemented using soft materials by employing nano-pumps and operating in higher frequency bands, the micro-fluidic system for the fluid antenna could have a diameter below $1$mm, and the response time is anticipated to fall within the sub-millisecond range \cite{convery201930}. Hence, the switching delay can be assumed to be negligible.} in a preset space. By doing so, FAS enables the mobile receiver's side\footnote{That said, the use of FAS is not limited to the UE side and some recent works have explored the advantage of FAS at base stations, see \cite{Wang-fas-isac2024}.} to obtain spatial diversity without the physical limitations of half-wavelength antenna spacing. This idea was motivated by the recent advances in flexible antennas such as liquid metal antennas or ionized solutions as well as reconfigurable pixel-like antennas, e.g., \cite{song2013efficient,convery201930,dey2016microfluidically,singh2019multistate}. For latest experimental results on FAS, readers are referred to \cite{Shen-tap_submit2024,Zhang-pFAS2024}.

\subsection{Related Works}
Several contributions have been recently made to investigate the performance of FAS in various wireless communication scenarios. In \cite{wong2020fluid}, the authors derived analytical expressions of the probability density function (PDF) and the cumulative density function (CDF) for their proposed FAS under spatially-correlated Rayleigh fading channels, and then obtained the exact and approximated outage probability (OP) in integral-form and closed-form expressions, respectively. The authors in \cite{tlebaldiyeva2022enhancing} recently derived an integral-form expression of the OP for a point-to-point (P2P) FAS under Nakagami-$m$ fading channels. In addition, the lower bound of ergodic capacity for FAS under Rayleigh fading channels was derived in \cite{wong2020performance}. Moreover, the OP for large-scale cellular networks that utilize FAS was analyzed in \cite{skouroumounis2022large}. Furthermore, in order to apply practical FAS to realistic wireless networks, the performance of multiuser communication systems exploiting fluid antennas was analyzed in \cite{wong2021fluid,wong2022fast,wong2023slow,10146262}. Additionally, the authors in \cite{chai2022port} proposed novel algorithms utilizing a combination of machine learning methods and analytical approximation to accurately select the best port with the maximum signal-to-noise ratio (SNR) in FAS when the system observes only a few ports. It is rightly understood that the FAS performance highly depends on the spatial correlation between the fluid antenna ports. Nevertheless, it was demonstrated in \cite{wong2022closed} that most previous investigations may not accurately capture such a correlation. Despite their proposed model allowing a tractable analysis, it is overly simplistic and degenerates Jakes' model correlation to a single averaged-correlation parameter; hence, the results are less accurate as reported in
further studies. In this regard, the authors in \cite{khammassi2023new} proposed an eigenvalue-based model to approximate the spatial correlation given by Jakes' model, where they indicated that under such model, the FAS has limited performance gain as the number of ports grows. Channel estimation for FAS was also addressed in \cite{Hao-2024}. For a comprehensive tutorial on FAS, readers can check \cite{New-submit2024}.

However, with all the recent studies and aforementioned considerations, there are several practical questions over the FAS to date: (\textit{i}) What is the correlation structure that needs to be designed to be able to maximize the performance of FAS? (\textit{ii}) How to build a tractable model that can work for any arbitrary correlation matrix and capture the case of Jakes' as an example? Therefore, a tractable statistical approach is required to gain more insight into these important issues.

\subsection{Motivation and Contributions}
As mentioned above, one of the most important challenges in channel modeling of FAS is to accurately characterize the spatial correlation between the fluid antenna ports --- and desirably, without a prohibitive complexity. Although great efforts have been performed in this context, there is a lack of precise methods that can describe such inherent correlation. Specifically, generating the joint multivariate distributions of correlated channels in FAS is extremely demanding due to mathematical and statistical limitations. For this reason, most previous works applied either the traditional statistical methods or asymptotic formulations to describe the  fading channel correlation between the fluid antenna ports. One flexible statistical procedure to overcome this issue is to adopt copula theory which has become popular recently in performance analysis of various wireless communication systems, e.g.,  \cite{ghadi2020copula,jorswieck2020copula,ghadi2020copula1,besser2021fading,ghadi2022capacity,besser2020bounds,ghadi2022performance,trigui2022copula,ghadi2022impact,besser2020reliability,ghadi2023ris}. 

In general, copulas are functions that can: $(i)$ Generate the joint multivariate distributions of two or more arbitrary random variables (RVs) by only knowing the marginal distributions; $(ii)$ describe the negative/positive dependence structure between two or more arbitrary  RVs beyond linear correlation. Exploiting such properties, the authors in \cite{ghadi2023copula} have recently analyzed the performance of FAS under arbitrary fading distributions, using the family of Archimedean copulas to derive the OP in a closed-form expression. However, this approach can only capture the impact of the number of fluid antenna ports on the OP performance, and cannot evaluate the effect of fluid antenna size on the FAS performance --- which is an arguably more important system parameter. 

Motivated by the aforesaid observations, this paper proposes a novel copula-based technique to investigate the performance of FAS under arbitrary fading distributions. By expressing the well-known Jakes' model in terms of Gaussian copula\footnote{Gaussian copula is an elliptical and symmetric copula that is determined entirely by its correlation matrix, instead of a single parameter as with the Archimedean copulas. In contrast to other copulas, the univariate margins in elliptical copula are joined by an elliptical distribution, which provides nice analytical properties \cite{frahm2003elliptical}.} to describe the spatial correlation amongst the fluid antenna ports, we first study a general case that is applicable for any arbitrary correlated fading distribution, and then analyze a specific case of correlated Nakagami-$m$ fading channels. In both cases, we derive compact analytical expressions of the PDF and CDF for the considered FAS, exploiting Gaussian copula. Then, to evaluate the system performance, we obtain the analytical expressions of the OP and delay outage rate (DOR) for both cases. To gain more insight into the superiority of our copula-based approach, we also measure the degree of dependence between correlated antenna ports. Specifically, the main contributions of our work are summarized as follows:
\begin{itemize}
\item We provide general formulations of the CDF, PDF, OP, and DOR in terms of the multivariate normal distribution for the considered FAS for any arbitrary fading distribution and also Jakes' model as an example.
\item In addition, we quantify the structure of dependency between spatially correlated antenna ports by exploiting the important rank correlation coefficients such as Spearman's $\rho$ and Kendall's $\tau$, where the scatterplots for the $2$-port FAS are provided for exemplary purposes. 
\item Furthermore, we derive the CDF, PDF, OP, and DOR in terms of the multivariate normal distribution under correlated Nakagami-$m$ fading channels.  
\item Numerical results show that Gaussian copula can accurately capture the spatial correlation between the fluid antenna ports in terms of the fluid antenna size and the number of ports. Moreover, the results reveal that the performance of FAS highly depends on the fluid antenna size and the number of ports. Increasing the fluid antenna size improves the OP and DOR but the performance does not necessarily get better as the number of ports grows.
\end{itemize}

\subsection{Paper Organization}
The rest of this paper is organized as follows. Section \ref{sec-sys} describes the system model. Section \ref{sec-per} presents the statistical characterization of the equivalent fading distribution at the FAS receiver end, with the fundamentals of copula theory and the general arbitrarily correlated fading distribution discussed in Section \ref{sub-copula} and Section \ref{sub-gen}, respectively. Section \ref{sec-perf} outlines the performance analysis of the considered FAS, where the analysis of the OP and DOR for the general case (i.e., arbitrary fading distribution) is provided in Section \ref{sub-gen2}, while the specific case (i.e., Nakagami-$m$) is discussed in Section \ref{sub-sp}. Then in Section \ref{sec-num}, the efficiency of analytical results is illustrated numerically, and finally, the conclusions are drawn in Section \ref{sec-con}.

\subsection{Mathematical Notation}
We use boldface upper and lower case letters for matrices and column vectors, respectively. $\mathrm{Cov}[\cdot]$ and $\mathrm{Var}[\cdot]$ denote the covariance and variance operators, respectively. Moreover, $(\cdot)^T$, $(\cdot)^{-1}$, $|.|$, and $\textrm{det}(\cdot)$ stand for the transpose, inverse, magnitude, and determinant, respectively.

\section{System Model}\label{sec-sys}
We consider a wireless communication system as shown in Fig.~\ref{fig-sys}, where a single-fixed-antenna transmitter sends an  independent message $x$ with transmit power $P$ to a mobile receiver that is equipped with a fluid antenna. It is assumed that the fluid antenna includes only one RF chain and $K$ preset positions (i.e., ports), which are equally distributed on a linear space of length $W\lambda$ where $\lambda$ denotes the wavelength of radiation in vacuum. Therefore, the distance between the first port and the $k$-th port is given by  
\begin{align}
d_k=\left(\frac{k-1}{K-1}\right)W\lambda, \quad \text{for}\, k=1,2,\dots,K.
\end{align}

\begin{figure}[!t]
\centering
`\includegraphics[width=0.9\columnwidth]{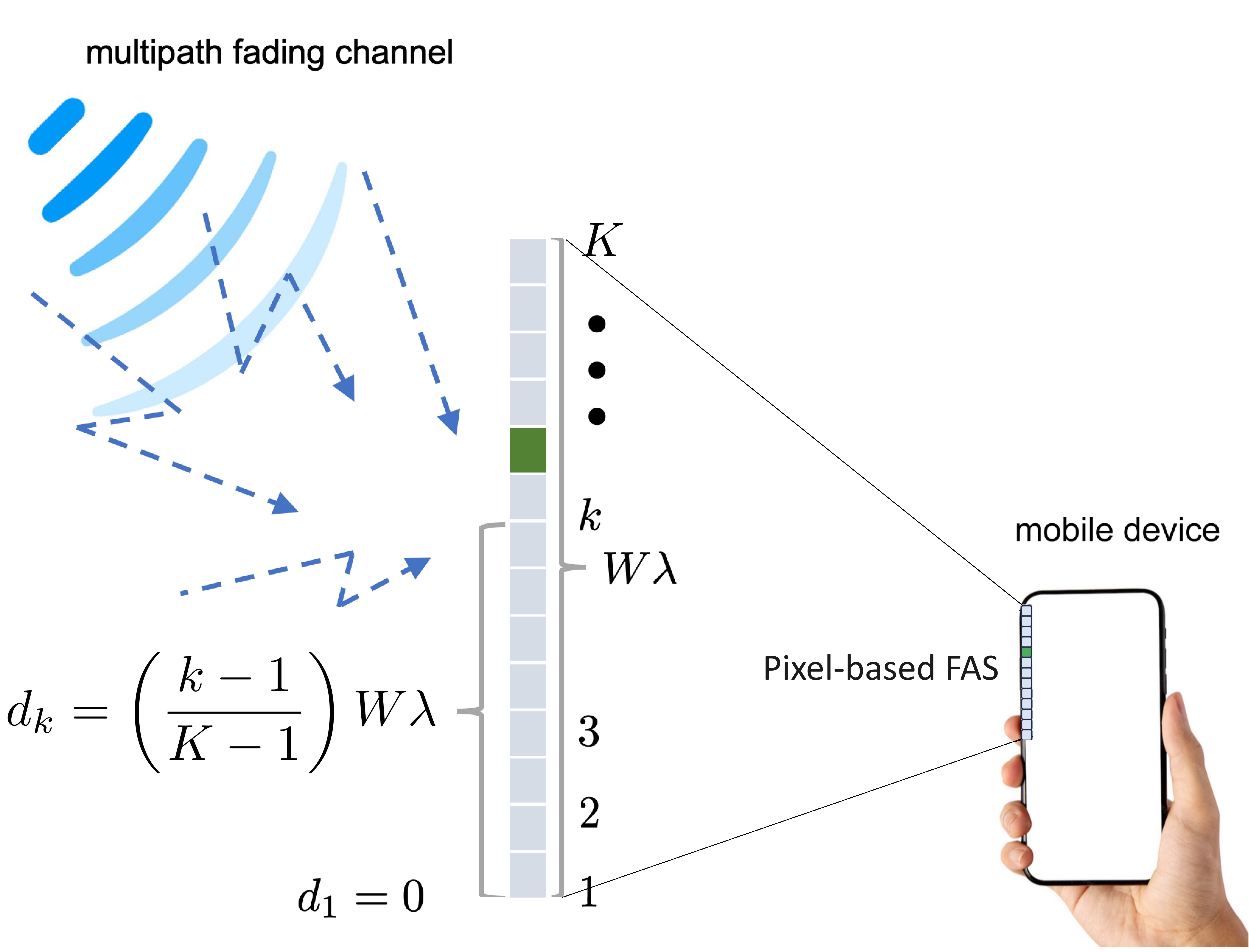}
\caption{
Exemplary illustration of FAS when the fluid antenna is implemented by reconfigurable pixels.}\label{fig-sys}
\end{figure}

Furthermore, the received signal at the $k$-th port under flat fading condition can be expressed as
\begin{align}
y_k=h_kx+z_k,
\end{align}
in which $h_k$ denotes the fading channel coefficient of the $k$-th port and $z_k$ is the independent identically distributed (i.i.d.) additive white Gaussian noise (AWGN) with zero mean and variance $\sigma_n^2$ at every port.

Given that the fluid antenna ports can be arbitrarily close to each other, the fading channel coefficients $\left\{h_k\right\}_{\forall k}$ are spatially correlated and have a covariance matrix $\vec{K}$. Assuming two-dimensional isotropic scattering (i.e., rich scattering) and isotropic receiver ports on the FAS, such spatial correlation can be characterized by Jakes' model so that \cite{stuber2001principles}
\begin{align}\label{eq-cov}
\vec{K}_{h_k,h_l}={\rm Cov}\left[h_k,h_l\right]=\sigma^2J_0\left(\frac{2\pi\left(k-l\right)}{K-1}W\right),
\end{align}
where $\sigma^2$ accounts for the large-scale fading effect and $J_0\left(\cdot\right)$ denotes the zero-order Bessel function of the first kind.

Moreover, in order to achieve optimal performance, FAS is assumed to always activate the best port with the maximum signal envelope for communication,\footnote{In a practical scenario, only a small subset of observed ports is needed to reach the full channel state information (CSI) \cite{chai2022port}.} i.e., 
\begin{align}
h_\mathrm{FAS}=\max\left\{|h_1|,|h_2|,\dots,|h_K|\right\}.
\end{align}  
Furthermore, the received SNR for the FAS can be defined as
\begin{align}
\gamma=\frac{Ph_\mathrm{FAS}^2}{\sigma_n^2}=\bar{\gamma}h_\mathrm{FAS}^2, \label{eq-snr}
\end{align}
in which $\bar{\gamma}=\frac{P^2}{\sigma_n^2}$ is the average transmit SNR. 

\section{Statistical Characterization} \label{sec-per}
In this section, we first characterize the spatial correlation between the fluid antenna ports in terms of Gaussian copula, and then derive compact analytical expressions of the channel distributions for a general case. Next, we measure the structure of dependency for the FAS  by exploiting rank correlation coefficients. To this end, we find it useful to first briefly review the concept of copula theory \cite{nelsen2007introduction}.  

\subsection{Brief Review of Copula Theory}\label{sub-copula}
\begin{definition}[$d$-dimensional copula]\label{def-copula}
Let $\vec{s}=[S_1,\dots,S_d]$ be a vector of $d$ RVs with marginal CDFs $F_{S_i}(s_i)$ for $i\in\{1,2,\dots,d\}$, respectively. Then, the corresponding joint CDF is defined as
\begin{align}
F_{S_1,\dots,S_d}(s_1,s_2,\dots,s_d)=\Pr(S_1\leq s_1,\dots,S_d\leq s_d).
\end{align}
The copula function $C(u_1,\dots,u_d)$ of the random vector $\vec{s}$ defined on the unit hypercube $[0,1]^d$ with uniformly distributed RVs $U_i:=F_{S_i}(s_i)$ over $[0,1]$ is given by
\begin{equation}
C(u_1,\dots,u_d)=\Pr(U_1\leq u_1,\dots,U_d\leq u_d),
\end{equation}
where $u_i=F_{S_i}(s_i)$.
\end{definition}

\begin{theorem}[Sklar's theorem]\label{thm-sklar}
Let $F_{S_1,\dots,S_d}(s_1,\dots,s_d)$ be a joint CDF of RVs with margins $F_{S_i}(s_i)$ for $i\in\{1,2,\dots,d\}$. Then, there exists one copula function $C$ such that for all $s_i$ in the extended real line domain $\bar{R}$,
\begin{equation}\label{eq-sklar}
F_{S_1,\dots,S_d}(s_1,\dots,s_d)=C\left(F_{S_1}(s_1),\dots,F_{S_d}(s_d)\right).
\end{equation}
\end{theorem}

\begin{definition}[$d$-dimension Gaussian copula]\label{def-gaussian} 
The multivariate Gaussian copula with correlation matrix $\vec{R}\in\left[-1,1\right]^{d\times d}$ is defined as
\begin{align}\label{eq-gaussian}
C_\mathrm{G}(u_1,\dots,u_d)=\Phi_\vec{R}\left(\phi^{-1}(u_1),\dots,\phi^{-1}(u_d);\eta\right),
\end{align}
where $\phi^{-1}(\cdot)$ is the inverse CDF (i.e., quantile function) of the standard normal distribution, $\Phi_\vec{R}(\cdot)$ is the joint CDF of the multivariate normal distribution with zero mean vector and correlation matrix $\vec{R}$, and $\eta$ denotes the dependence parameter of the Gaussian copula which can control the degree of dependence between correlated RVs.
\end{definition}

\subsection{General Case: Arbitrary Correlated Fading Distribution}\label{sub-gen}
Classically, spatial correlation in wireless communications is modeled through linear correlation; e.g., Jakes' correlation model is assumed in FAS. Unfortunately, there is no analytical solution for the received signal in FAS under such model. Evidently, system performance needs to be analyzed over correlated fading channels due to the close proximity of the FAS ports. In this line, while the linear correlation coefficient is effective in determining elliptical multivariate distributions, it sometimes falls short in capturing the correlations present in non-elliptical multivariate distributions like Rayleigh or Rician fading models since such distributions, being Gaussian-centric, may not adequately represent simultaneous deep fades influenced by underlying interdependencies. Moreover, the linear correlation parameter demonstrates satisfactory performance for the majority of the practical $d$-variate Nakagami-$m$ distribution; however, its approximation often tends to falter in the tails, which is momentous because bit errors or outages mainly occur in deep fade conditions. As a consequence, there is an increased demand for a more general mathematical and statistical tool that can accurately model the correlated fading channel; especially, in FAS. For this purpose, the copula theory can be considered a tractable approach to address the challenges related to the structure of dependency in FAS.

Now, choosing the optimum copula that can exactly fit with the pre-defined system model and also accurately measure the unknown dependence structure is challenging. In this regard, although the empirical method is often accurate and provides realistic observations, it needs real datasets for such analysis \cite{ruschendorf2009distributional}, which are usually unavailable; especially for FAS-based communications. Moreover, while Archimedean copulas have nice properties and a simple structure that can accurately describe the structure of dependency, especially the lower and upper tail dependencies, they do not include the correlation/covariance matrix and their corresponding rank correlation cannot be easily derived \cite{nelsen1997dependence}. In contrast, the Gaussian copula includes a correlation matrix such that each entry of the matrix is the dependence parameter of the Gaussian copula, i.e., $\eta$. As we later show in detail, $\eta$ can directly approximate the correlation provided by Jakes' model by a simple transformation of the rank correlations. Additionally, the Gaussian copula can be extended to incorporate tail dependence, allowing for the modeling of extreme events and dependencies in the tails of the distribution \cite{frahm2003elliptical}, which is particularly important in the performance analysis of wireless communication systems. Hence, this choice is valid for the main purposes of this paper.
\subsubsection{Channel distributions} 
By exploiting the concept of copula theory and using the definition of $d$-dimension Gaussian copula, we derive the compact analytical expressions of the CDF and  PDF of $h_\mathrm{FAS}$ in the following theorems. 

\begin{theorem}\label{thm-gen-cdf}
The CDF of $h_\mathrm{FAS}=\max\left\{|h_1|,|h_2|,\dots,|h_K|\right\}$ for any arbitrary correlated fading coefficient $|h_k|$, $k\in\{1,2,\dots,K\}$, with marginal CDF $F_{|h_k|}\left(r\right)$ by exploiting Gaussian copula is derived as
\begin{align}\label{eq-cdf-gaussian}
F_{h_\mathrm{FAS}}\left(r\right)=\Phi_{\vec{R}_{h_k,h_l}}\left(\phi^{-1}\left(F_{|h_1|}\left(r\right)\right),\dots,\phi^{-1}\left(F_{|h_K|}\left(r\right)\right)\right),
\end{align}
where 
\begin{align}\label{eq-phi-inv}
\phi^{-1}\left(F_{|h_K|}\left(r\right)\right)=\sqrt{2}\,\mathrm{erf}^{-1}\left(2F_{|h_K|}\left(r\right)-1\right),
\end{align}
in which $\mathrm{erf}^{-1}(\cdot)$ denotes the inverse of error function $\mathrm{erf}(z)=\frac{2}{\sqrt{\pi}}\int_{0}^z\mathrm{e}^{-t^2}\mathrm{d}t$. The term $\Phi_{\vec{R}_{h_k,h_l}}(\cdot)$ is the joint CDF of the multivariate normal distribution with zero mean vector and the correlation matrix $\vec{R}_{h_k,h_l}$ as
\begin{align}\label{eq-corr-matrix}
\vec{R}_{h_k,h_l}=\begin{bmatrix}
1 & \eta_{1,2} & \hdots & \eta_{1,l}\\
\eta_{2,1} & 1 & \hdots & \eta_{2,l}\\
\vdots & \vdots & \ddots& \vdots\\
\eta_{k,1} & \eta_{k,2} & \hdots& 1\\
\end{bmatrix},
\end{align}
where $\eta_{k,l}=J_0\left(\frac{2\pi\left(k-l\right)}{K-1}W\right)$ is the dependence parameter of Gaussian copula which can be chosen freely to control the correlation between the corresponding channel coefficients.
\end{theorem}

\begin{proof}
Since the structure of dependency between correlated RVs in the Gaussian copula is denoted by a correlation matrix with corresponding dependence parameters, we first determine the correlation matrix between the arbitrary channel coefficients in terms of Jakes' model. Hence, by exploiting the covariance matrix in \eqref{eq-cov} and considering Cholesky decomposition, the correlation matrix $\vec{R}_{h_k,h_l}$ that includes dependence parameter $\eta_{k,l}$ is obtained as
\begin{align}\label{eq-corr}
\vec{R}_{h_k,h_l}=\frac{{\rm Cov}\left[h_k,h_l\right]}{\sqrt{{\rm Var}\left[h_k\right]{\rm Var}\left[h_l\right]}}=J_0\left(\frac{2\pi\left(k-l\right)}{K-1}W\right).
\end{align}
Next, by using the definition of the CDF, $F_{h_\mathrm{FAS}}\left(r\right)$ can be mathematically expressed as
\begin{align}\label{eq-proof-cdf}
F_{h_{\mathrm{FAS}}}(r)
&\hspace{.5mm}=\Pr\left(\max\left\{|h_1|,|h_2|,\dots,|h_K|\right\}\leq r\right)\\
&\hspace{.5mm}=\Pr\left(|h_1|\leq  r,|h_2|\leq r,\dots,|h_K|\leq r\right)\\
&\hspace{.5mm}=F_{|h_1|,|h_2|,\dots,|h_K|}\left(r,r,\dots,r\right)\\
&\overset{(a)}{=}C\left(F_{|h_1|}(r),F_{|h_2|}(r),\dots,F_{|h_K|}(r)\right),
\end{align}
where $(a)$ is obtained from Theorem \ref{thm-sklar}. Now, by inserting the Gaussian copula from \eqref{eq-gaussian} into \eqref{eq-proof-cdf}, then plugging the marginal CDF of arbitrary fading channels into the obtained result for $u_d=F_{|h_k|}(r)$, and finally considering the correlation matrix in \eqref{eq-corr}, the proof is completed. 
\end{proof}

\begin{theorem}\label{thm-gen-pdf}
The PDF of $h_\mathrm{FAS}=\max\left\{|h_1|,|h_2|,\dots,|h_K|\right\}$ for any arbitrary correlated fading coefficient $|h_k|$, for $k\in\{1,2,\dots,K\}$, with marginal PDF $f_{|h_k|}\left(r\right)$ and marginal CDF $F_{|h_k|}\left(r\right)$ by exploiting Gaussian copula is derived as
\begin{multline}\label{eq-pdf-gaussian}
f_{h_\mathrm{FAS}}\left(r\right)=\prod_{k=1}^Kf_{|h_k|}\left(r\right)\\
\times\frac{\exp\left(-\frac{1}{2}\left(\boldsymbol{\phi}^{-1}_{h_K}\right)^T\left(\vec{R}^{-1}_{h_k,h_l}-\vec{I}\right)\boldsymbol{\phi}^{-1}_{h_K}\right)}{\sqrt{{\rm det}\left(\vec{R}_{h_k,h_l}\right)}},
\end{multline}
where ${\rm det}\left(\vec{R}_{h_k,h_l}\right)$ denotes the determinant of the correlation matrix $\vec{R}_{h_k,h_l}$, $\vec{I}$ is the identity matrix, and  $\boldsymbol{\phi}^{-1}_{h_K}=\left[\phi^{-1}\left(F_{|h_1|}(r)\right),\dots,\phi^{-1}\left(F_{|h_K|}(r)\right)\right]^T$ is given by \eqref{eq-phi-inv}.
\end{theorem}

\begin{proof}
By applying the chain rule to $F_{h_\mathrm{FAS}}\left(r\right)$ provided in \eqref{eq-cdf-gaussian}, and then considering the marginal distributions of arbitrary fading channels, we have
\begin{align}\nonumber
&f_{h_{\mathrm{FAS}}}\left(r\right)=\prod_{k=1}^K f_{|h_k|}\left(r\right)\\
&\hspace{-0.1cm}\times\underset{c_{\mathrm{G}}\left(F_{|h_1|}\left(r\right)\right),\dots,\left(F_{|h_K|}\left(r\right)\right)}{\underbrace{\frac{\partial^K \Phi_{\vec{R}_{h_k,h_l}}\left(\phi^{-1}\left(F_{|h_1|}\left(r\right)\right),\dots,\phi^{-1}\left(F_{|h_K|}\left(r\right)\right)\right)}{\partial F_{|h_1|}(r)\dots \partial F_{|h_K|}(r)}}},\label{eq-def-pdf}
\end{align}
where $c_{\mathrm{G}}\left(F_{|h_1|}\left(r\right)\right),\dots,\left(F_{|h_K|}\left(r\right)\right)$ is the Gaussian copula density which can be computed mathematically as
\begin{multline}\label{eq-den-gaussian}
c_{\mathrm{G}}\left(F_{|h_1|}\left(r\right)\right),\dots,\left(F_{|h_K|}\left(r\right)\right)\\
=\frac{\exp\left(-\frac{1}{2}\left(\boldsymbol{\phi}^{-1}_{h_K}\right)^T\left(\vec{R}^{-1}_{h_k,h_l}-\vec{I}\right)\boldsymbol{\phi}^{-1}_{h_K}\right)}{\sqrt{{\rm det}\left(\vec{R}_{h_k,h_l}\right)}}.
\end{multline}
Now, by inserting \eqref{eq-den-gaussian} into \eqref{eq-def-pdf} and considering the correlation matrix $\vec{R}_{h_k,h_l}$, the proof is completed. 
\end{proof}

\begin{figure*}[t!]
\centering
\subfigure[$W=0.05$]{%
\includegraphics[width=0.26\textwidth]{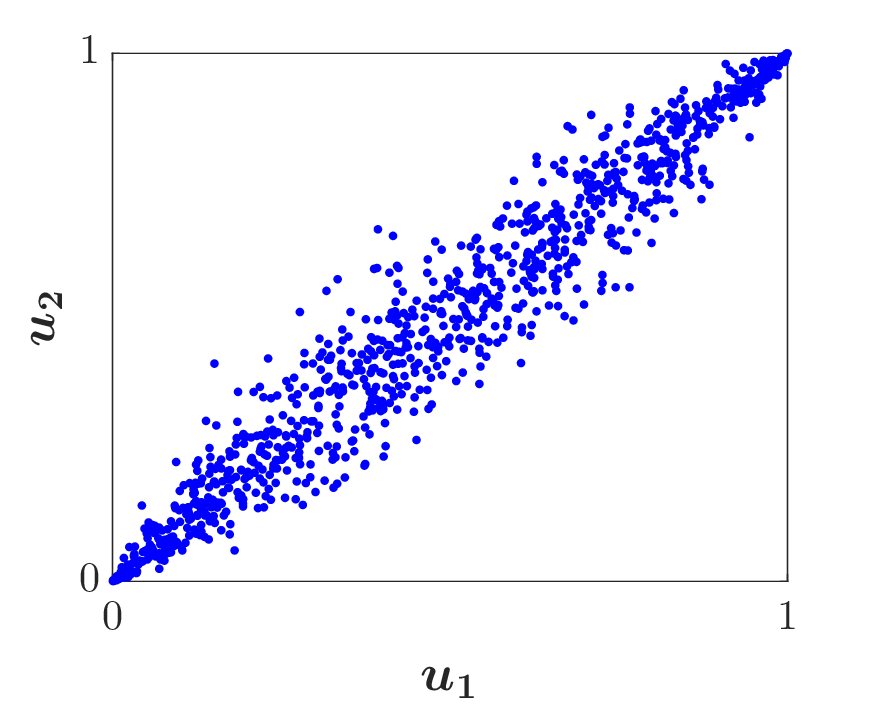}\label{}%
}\hspace{-0.37cm}
\subfigure[$W=0.1$]{%
\includegraphics[width=0.26\textwidth]{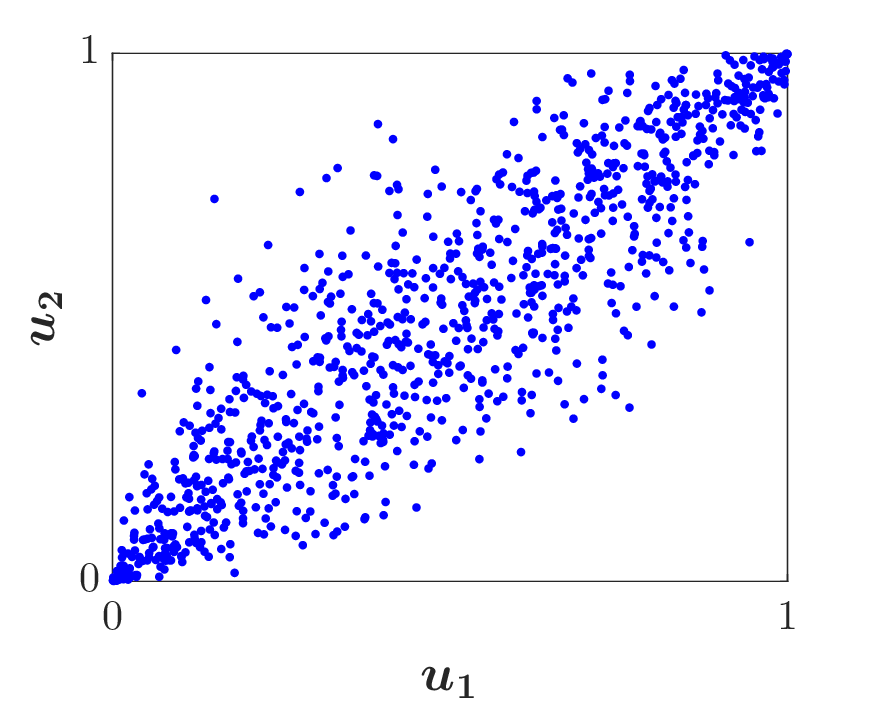}\label{}%
}\hspace{-0.37cm}
\subfigure[$W=0.5$]{%
\includegraphics[width=0.26\textwidth]{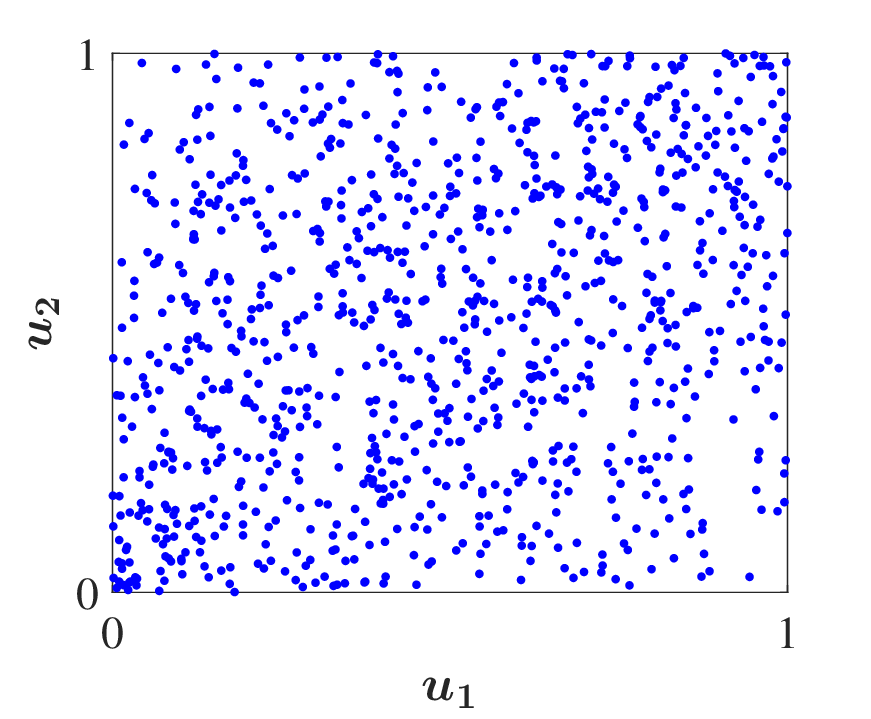}\label{}%
}\hspace{-0.37cm}%or more
\subfigure[$W=4$]{%
\includegraphics[width=0.26\textwidth]{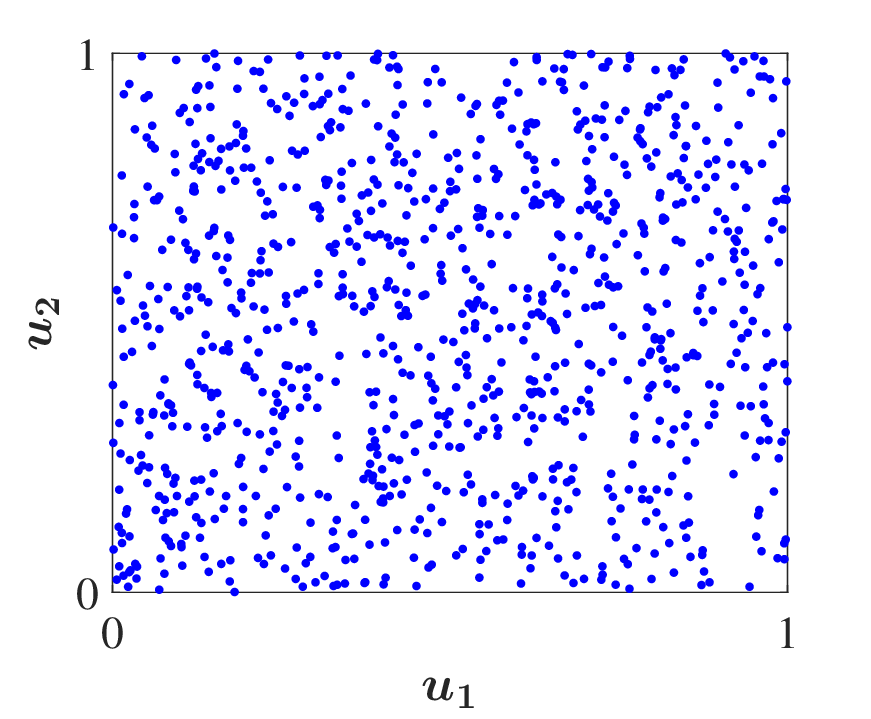}\label{}%
}% 
\caption{Scatterplots describe the structure of dependency between two arbitrary correlated fading channels $|h_1|$ and $|h_2|$ with uniform marginal distributions $u_1$ and $u_2$ under Gaussian copula that includes correlation matrix $\vec{R}_{h_1,h_2}$.}\label{fig-scatter-u}\vspace{0cm}
\end{figure*}

\subsubsection{Dependence measurement} 
In order to gain more insight into how our copula-based analytical results can describe the structure of dependency between correlated fading channel coefficients over the fluid antenna ports, we here provide the scatterplots of two arbitrarily correlated fading channels $|h_1|$ and $|h_2|$ by using Gaussian copula over the $2$-port FAS in Fig.~\ref{fig-scatter-u}. To do so, by considering uniformly distributed RVs $u_1=F_{|h_1|}\left(|h_1|\right)$ and $u_2=F_{|h_2|}\left(|h_2|\right)$, we generate $n=1000$ random vectors from the Gaussian copula $C_\mathrm{G}\left(u_1,u_2\right)$ with correlation matrix $\vec{R}_{h_1,h_2}$ consisting of two dependence parameters $\eta_{1,2}$ and $\eta_{2,1}$. Given that we considered Jakes' model to describe the spatial correlation between the fluid antenna ports, both dependence parameters highly depend on the fluid antenna size $W$. Thus, we can observe from Fig.~\ref{fig-scatter-u} that as the size of the fluid antenna increases (decreases), the scattering between data generated by the Gaussian copula increases (decreases), meaning that the spatial correlation between the fluid antenna ports becomes weaker (stronger). Furthermore, these observations are completely separate from the marginal distributions of channel coefficients, which means that $|h_1|$ and $|h_2|$ can be modeled by any arbitrary marginal distributions whereas they still have the same \textit{rank} correlation. However, to realize the spatial correlation between the fluid antenna ports under specific channel distributions, we just need to apply the inverse CDF of the corresponding distribution to the uniformly distributed RVs $u_1$ and $u_2$ (i.e., $|h_1|=F^{-1}_{|h_1|}\left(u_1\right)$ and  $|h_2|=F^{-1}_{|h_2|}\left(u_2\right)$), in which the spatial correlation for a specific scenario will be analyzed in the next section. 

On the other hand, it should be noted that the dependence parameter of Gaussian copula $\eta$ in Def. \ref{def-gaussian} does not necessarily express the linear correlation between two or more RVs, since when non-linear transformations are applied to those RVs, the linear correlation cannot be maintained anymore (as it happens in the copula definition). For this reason, a rank correlation that can measure the statistical dependence between the \textit{ranking} of two RVs is often used to describe the correlation between RVs. In other words, unlike the linear correlation coefficient which can only describe linear dependence between RVs, rank correlations are preserved under any monotonic transformation, and thus can significantly describe the non-linear correlation between RVs. In this regard, Spearman's $\rho$, denoted as $\rho_\mathrm{s}$, and Kendall's $\tau$, denoted as $\tau_\mathrm{k}$, are the two most popular rank correlation coefficients which can accurately describe such dependence between RVs since they are invariant to the choice of the marginal distribution. Additionally, $\rho_\mathrm{s}$ and $\tau_\mathrm{k}$  for the two arbitrary RVs with the corresponding copula $C$ can be expressed, respectively, as \cite{nelsen2007introduction}
\begin{align}
\rho_\mathrm{s}&=12\iint_{\left[0,1\right]^2}u_1u_2dC\left(u_1,u_2\right)-3,\label{eq-def-rho-s}\\
\tau_\mathrm{k}&=4\iint_{\left[0,1\right]^2}C\left(u_1,u_2\right)dC\left(u_1,u_2\right)-1.\label{eq-def-tau}
\end{align}

In particular, for the two arbitrary correlated fading channel coefficients $|h_1|$ and $|h_2|$ with the Gaussian copula $C_\mathrm{G}$, which is defined in terms of bivariate normal distribution, $\rho_{\mathrm{s}_{k,l}}$ and $\tau_{\mathrm{k}_{k,l}}$ can be, respectively, expressed in terms of $\eta_{k,l}$ as
\begin{align}
\rho_{\mathrm{s}_{k,l}}&=\frac{6}{\pi}\arcsin\left(\frac{\eta_{k,l}}{2}\right),\label{eq-rho-s-u}\\
\tau_{\mathrm{k}_{k,l}}&=\frac{2}{\pi}\arcsin\left(\eta_{k,l}\right).\label{eq-tau-u}
\end{align}
Therefore, by inserting the correlation matrix $\vec{R}_{h_k,h_l}$ from \eqref{eq-corr} which depends on the fluid antenna size $W$ into \eqref{eq-tau-u} and \eqref{eq-rho-s-u}, we provide the dependence measurement between two arbitrary correlated channel coefficients $|h_1|$ and $|h_2|$ over the FAS in terms of the rank correlations $\rho_{\mathrm{s}_{k,l}}$ and $\tau_{\mathrm{k}_{k,l}}$ in Tab.~\ref{tab-dependence}. It is worth noting that the rank correlation coefficients are strictly less than the Gaussian dependence parameter unless $\eta_{k,l}$ is exactly one (i.e., full correlation). 
\begin{remark}
It is worth mentioning that in \eqref{eq-corr} each entry of the covariance matrix in Jakes' model (i.e., correlation parameter) is connected to each entry of the correlation matrix of the Gaussian copula (i.e., the Gaussian copula parameter $\eta_{k,l}$ that controls the dependency) with the help of the Cholesky decomposition. Then, it is necessary to check how the underlying copula (e.g., the Gaussian copula in this paper) can be accurately matched with the given arbitrary marginal distributions. Despite several various approaches exist for this purpose \cite{xiao2019matching}, one of the main  methods is to describe the structure of dependency based on rank correlation coefficients such as Spearman's $\rho$ and Kendall's $\tau$ for the considered copula, as they are provided in \eqref{eq-rho-s-u} and \eqref{eq-tau-u} for the Gaussian copula. While this technique is mathematically based, it can provide useful insights into the behavior of the dependence structure. By examining the results in Table \ref{tab-dependence}, the accuracy of the Gaussian copula in describing the correlation between the channel gains in FAS can be observed. This will be validated through scatterplot analysis under the Gaussian copula and Monte-Carlo simulations using Jakes' model later in Sect. \ref{sec-perf}, for the exemplary case of Nakagami-$m$ fading. 
\end{remark}
\begin{remark}
In comparison with the previous contributions that exploited the Archimedean copulas in FAS \cite{ghadi2023copula},  \cite{ghadi2023fluid}, the analytical relation between rank correlation coefficients and the copula dependence parameter is less complicated under the Gaussian copula. For instance, in \cite{ghadi2023fluid}, an estimation between Spearman's $\rho$ and the dependence parameter of the Clayton copula was used to approximate Jakes' model, which makes the results less accurate. In contrast, the relation between $\rho_\mathrm{s}$ and the Gaussian copula parameter, which is provided in \eqref{eq-rho-s-u} is exact; thereby, Jakes' model is approximated more accurately in this paper.
\end{remark}
\begin{remark}\label{remark-complex}
	The derivation of the distribution of the maximum of $K$ correlated random variables requires for the evaluation of $K$-fold nested integrals \cite{zhang2002general}, which is prohibitively complex specially as the FAS size grows. Although complexity can be reduced if the equivalent distribution of the SNR is formulated using an $\epsilon$-rank approximation for the relevant set of eigenvalues, a multi-fold nested integration of $\epsilon$-rank order over Marcum Q-functions is still required even for the case of Rayleigh fading \cite{khammassi2023new}. Compared to these alternatives, the Gaussian copula approach provides an alternative solution that requires the evaluation of the multivariate normal CDF with arguments computed using the inverse error function; both are available in commercial software packages like Matlab through \texttt{mvncdf} and \texttt{erfinv}, respectively. Besides, the arbitrary choice of underlying fading distribution (e.g., Nakagami-$m$ in the example here considered) is directly captured through its marginal CDF in the argument of the \texttt{erfinv} evaluations. Since the evaluation of nested integrals over sophisticated special functions is not required, complexity is notably reduced when using our approach.
\end{remark}
\begin{table}\caption{Dependence measurement in terms of the fluid antenna size $W$ for the considered FAS when $K=2$}\centering
\begin{tabular}{ p{1.3cm}||p{0.5cm}|p{0.5cm}||p{0.5cm}|p{0.5cm}||p{0.5cm}|p{0.5cm} }
\hline
\multicolumn{7}{c}{Dependence measurement for a $2$-port FAS}\\
\hline
\centering Size $W$& $\eta_{1,2}$ &$\eta_{2,1}$&$\rho_{\mathrm{s}_{1,2}}$&$\rho_{\mathrm{s}_{2,1}}$&$\tau_{\mathrm{k}_{1,2}}$&$\tau_{\mathrm{k}_{2,1}}$\\
\hline
$W=0.05$   & $0.98$    &$0.98$&   $0.97$ & $0.97$ & $0.86$ & $0.86$\\
$W=0.1$&   $0.90$  & $0.90$   &$0.89$ & $0.89$ &$0.72$& $0.72$\\
$W=0.5$&   $0.30$  & $0.30$   &$0.29$ & $0.29$ &$0.20$& $0.20$\\
$W=1$ &$0.22$ & $0.22$&  $0.21$&$0.21$ &$0.14$& $0.14$\\
$W=2$    &$0.16$ & $0.16$&  $0.15$ & $0.15$ & $0.10$ & $0.10$\\
$W=4$&   $0.11$  & $0.11$ & $0.10$ & $0.10$ &$0.07$&$0.07$\\
$W=6$&   $0.09$  & $0.09$ & $0.09$ & $0.09$ &$0.06$ & $0.06$\\
\hline
\end{tabular}\label{tab-dependence}
\end{table}

\section{Performance Analysis}  \label{sec-perf}
Here, we first derive compact analytical expressions of the OP and DOR for a general case. Then in order to evaluate the FAS performance, we obtain the OP and DOR under correlated Nakagami-$m$ fading channels as a special case.

\subsection{General Case: Arbitrary Correlated Fading Distribution}\label{sub-gen2}
\subsubsection{OP analysis} 
OP is a key performance metric in wireless communication systems which is defined as the probability that the random SNR $\gamma$ is less than an SNR threshold $\gamma_\mathrm{th}$, i.e., $P_\mathrm{out}=\Pr\left(\gamma\leq\gamma_\mathrm{th}\right)$. Hence, we derive the OP for the considered FAS in the following theorem. 

\begin{theorem}\label{thm-gen-out}
The OP for the considered FAS under arbitrary correlated fading channels by exploiting Gaussian copula is given by
\begin{align}
P_{\mathrm{out}}=\Phi_{\vec{R}_{h_k,h_l}}\left(\phi^{-1}\left(F_{|h_1|}\left(\hat{\gamma}\right)\right),\dots,\phi^{-1}\left(F_{|h_K|}\left(\hat{\gamma}\right)\right)\right),
\end{align}
in which $\phi^{-1}\left(F_{|h_K|}\left(\hat{\gamma}\right)\right)=\sqrt{2}\,\mathrm{erf}^{-1}\left(2F_{|h_K|}\left(\hat{\gamma}\right)-1\right)$, $\hat{\gamma}=\sqrt{\frac{\gamma_\mathrm{th}}{\bar{\gamma}}}$, and  the correlation matrix $\vec{R}_{h_k,h_l}$ is defined in \eqref{eq-corr-matrix}. 
\end{theorem}

\begin{proof}
By substituting the SNR of the FAS from \eqref{eq-snr} into the definition of OP, we have
\begin{align}
P_\mathrm{out}=\Pr\left(\max\left\{|h_1|,\dots,|h_K|\right\}\leq\sqrt{\frac{\gamma_\mathrm{th}}{\bar{\gamma}}}\right)=F_{h_\mathrm{FAS}}\left(\hat{\gamma}\right),
\end{align}
where by utilizing the CDF in \eqref{eq-cdf-gaussian}, the proof is completed.
\end{proof}

\subsubsection{DOR analysis}
DOR is an important performance metric for designing ultra-reliable and low-latency communications (URLLC), which is defined as the probability that the time required to successfully transmit a certain amount of data $R$ in a wireless channel with a bandwidth $B$ is greater than a threshold duration $T_\mathrm{th}$, i.e., $P_\mathrm{dor}=\Pr\left(T_\mathrm{dt}>T_\mathrm{th}\right)$, where  
\begin{equation}
T_\mathrm{dt}=\frac{R}{B\log_2\left(1+\gamma\right)} 
\end{equation}
represents the delivery time \cite{yang2019ultra}. Hence, the DOR for the considered FAS can be derived in the following theorem. 

\begin{theorem}\label{thm-gen-dor}
The DOR for the considered FAS under arbitrary correlated fading channels by exploiting Gaussian copula is given by
\begin{align}
P_{\mathrm{dor}}=\Phi_{\vec{R}_{h_k,h_l}}\left(\phi^{-1}\left(F_{|h_1|}\left(\hat{T}\right)\right),\dots,\phi^{-1}\left(F_{|h_K|}\left(\hat{T}\right)\right)\right),
\end{align}
where $\phi^{-1}\left(F_{|h_K|}\left(\hat{T}\right)\right)=\sqrt{2}\,\mathrm{erf}^{-1}\left(2F_{|h_K|}\left(\hat{T}\right)-1\right)$, $\hat{T}=\sqrt{\frac{\mathrm{e}^{\frac{R\ln2}{BT_\mathrm{th}}}}{\bar{\gamma}}}$, and $\vec{R}_{h_k,h_l}$ is obtained from \eqref{eq-corr-matrix}. 
\end{theorem}

\begin{proof}
By applying the SNR of the FAS to the definition of the DOR, we have 
\begin{align}
P_\mathrm{dor}&=\Pr\left(\frac{R}{B\log_2\left(1+\gamma\right)}> T_\mathrm{th}\right)\\
&=\Pr\left(\gamma\leq\exp\left(\frac{R\ln2}{BT_\mathrm{th}}\right)-1\right)\\
&=\Pr\left(h_\mathrm{FAS}\leq\sqrt{\frac{\mathrm{e}^{\frac{R\ln2}{BT_\mathrm{th}}}-1}{\bar{\gamma}}}\right)\\
&=F_\mathrm{FAS}\left(\hat{T}\right),
\end{align}
which after using \eqref{eq-cdf-gaussian}, completes the proof.
\end{proof}

\begin{remark}
It is worth noting that the obtained analytical results in Theorems~\ref{thm-gen-cdf}, \ref{thm-gen-pdf}, \ref{thm-gen-out}, \ref{thm-gen-dor} are valid for any choice of arbitrary correlated fading distributions. Though these analytical results, except the PDF, are expressed in terms of the CDF of multivariate normal distributions, they can be estimated numerically by adopting different methods \cite{genz1999numerical,genz2002comparison1,wei2021approximation} and there is no need to solve any complicated integral. In addition, in contrast to \cite{ghadi2023copula} which did not consider the size of the fluid antenna for describing the channel correlation,  we can see that our copula-based analytical results can accurately describe the spatial correlation between the fluid antenna ports in terms of Jakes' model, and hence, we can consider the fluid antenna size in the performance analysis of the FAS.
\end{remark} 
  
\subsection{Special Case: Correlated Nakagami-$m$ Fading}\label{sub-sp}
To analyze the efficiency of the considered FAS in terms of the OP and DOR under typical fading conditions, here we consider that the channel coefficients follow Nakagami-$m$ distribution, where the parameter $m\ge 0.5$ describes the fading severity. The Nakagami-$m$ distribution is chosen due to its versatility in representing different fading conditions through the parameter $m$ (e.g., Rayleigh distribution when $m=1$), and its relevance in correlated scenarios where spatial correlation arises naturally from the FAS. Therefore, the marginal PDF and CDF of the fading channel coefficient $|h_k|$ with the shape parameter $m$ and the spread parameter $\mu$ can be, respectively, written as
\begin{align}
f^{\mathrm{Nak}}_{|h_k|}(r)&=\frac{2m^m}{\Gamma(m)\mu^m}r^{2m-1}\mathrm{e}^{-\frac{m}{\mu}r^2},\label{eq-pdf-nak}\\
F^{\mathrm{Nak}}_{|h_k|}(r)&=\frac{\gamma\left(m,\frac{m}{\mu}r^2\right)}{\Gamma(m)},\label{eq-cdf-nak}
\end{align}
in which the terms $\Gamma(\cdot)$ and $\gamma(\cdot,\cdot)$ are gamma function and the lower incomplete gamma function, respectively.

\begin{corollary}
The CDF of $h_\mathrm{FAS}=\max\left\{|h_1|,|h_2|,\dots,|h_K|\right\}$ under correlated Nakagami-$m$ fading coefficient $|h_k|$, for $k\in\{1,2,\dots,K\}$, with marginal CDF $F^{\mathrm{Nak}}_{|h_k|}\left(r\right)$ by utilizing Gaussian copula is derived as
\begin{multline}\label{eq-cdf-g-nak}
F_{h_\mathrm{FAS}}^\mathrm{Nak}\left(r\right)=\Phi_{\vec{R}_{h_k,h_l}}\left(\sqrt{2}\,\mathrm{erf}^{-1}\left(\frac{2\gamma\left(m,\frac{m}{\mu}r^2\right)}{\Gamma(m)}-1\right)\right.\\
\left.,\dots,\sqrt{2}\,\mathrm{erf}^{-1}\left(\frac{2\gamma\left(m,\frac{m}{\mu}r^2\right)}{\Gamma(m)}-1\right)\right).
\end{multline}
\end{corollary}

\begin{proof}
By inserting \eqref{eq-cdf-nak} into \eqref{eq-phi-inv} and plugging the obtained result into \eqref{eq-cdf-gaussian}, the proof is completed. 
\end{proof}

\begin{corollary}
The PDF of $h_\mathrm{FAS}=\max\left\{|h_1|,|h_2|,\dots,|h_K|\right\}$ under correlated Nakagami-$m$ fading coefficient $|h_k|$, for $k\in\{1,2,\dots,K\}$, with marginal PDF $f^\mathrm{Nak}_{|h_k|}\left(r\right)$ and marginal CDF  $F^\mathrm{Nak}_{|h_k|}\left(r\right)$ by exploiting Gaussian copula is derived as
\begin{multline}
f^\mathrm{Nak}_{h_\mathrm{FAS}}\left(r\right)=\frac{\left(\frac{2m^m}{\Gamma(m)\mu^m}r^{2m-1}\mathrm{e}^{-\frac{m}{\mu}r^2}\right)^K}{\sqrt{\mathrm{det}\left(\vec{R}_{h_k,h_l}\right)}}\\
\times\exp\left(-\frac{1}{2}\left({\boldsymbol{\phi}^\mathrm{Nak}_{h_K}}^{-1}\right)^T\left(\vec{R}^{-1}_{h_k,h_l}-\vec{I}\right){\boldsymbol{\phi}^{\mathrm{Nak}}_{h_K}}^{-1}\right),
\end{multline}
where
\begin{multline}
{\boldsymbol{\phi}^{\mathrm{Nak}}_{h_K}}^{-1}=\left[\sqrt{2}\,\mathrm{erf}^{-1}\left(\frac{2\gamma\left(m,\frac{m}{\mu}r^2\right)}{\Gamma(m)}-1\right)\right.\\
\left.,\dots,\sqrt{2}\,\mathrm{erf}^{-1}\left(\frac{2\gamma\left(m,\frac{m}{\mu}r^2\right)}{\Gamma(m)}-1\right)\right].
\end{multline}
\end{corollary}

\begin{proof}
By substituting the marginal distribution of the Nakagami-$m$ channel $|h_K|$ from \eqref{eq-pdf-nak} and \eqref{eq-cdf-nak} into \eqref{eq-pdf-gaussian} and considering the correlation in \eqref{eq-corr-matrix}, the proof is completed. 
\end{proof}

\begin{figure*}[t!]
\centering
\subfigure[$W=0.05$]{%
\includegraphics[width=0.26\textwidth]{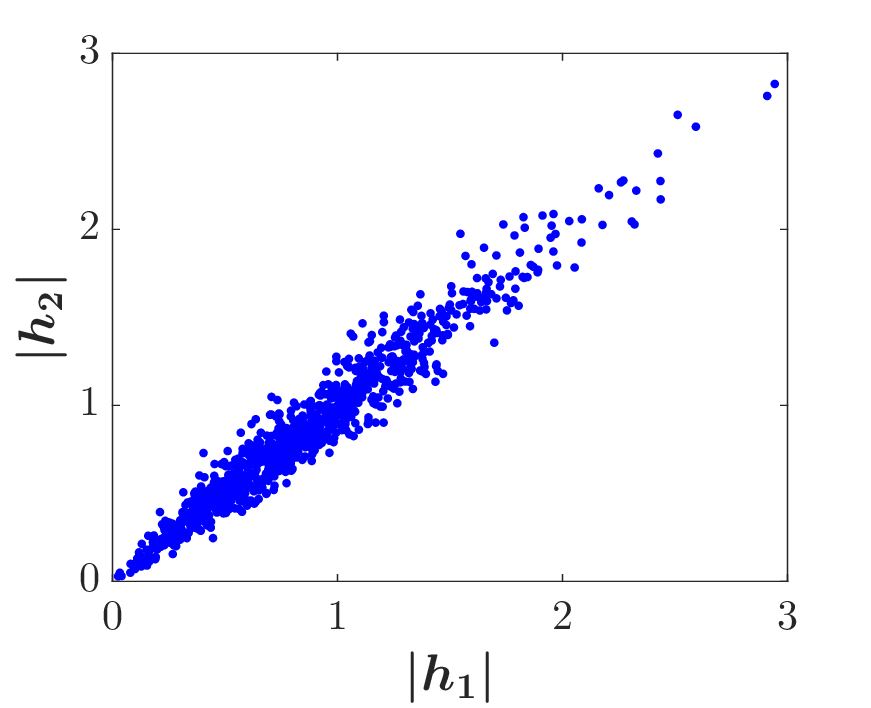}\label{sub-h1}%
}\hspace{-0.37cm}
\subfigure[$W=0.1$]{%
\includegraphics[width=0.26\textwidth]{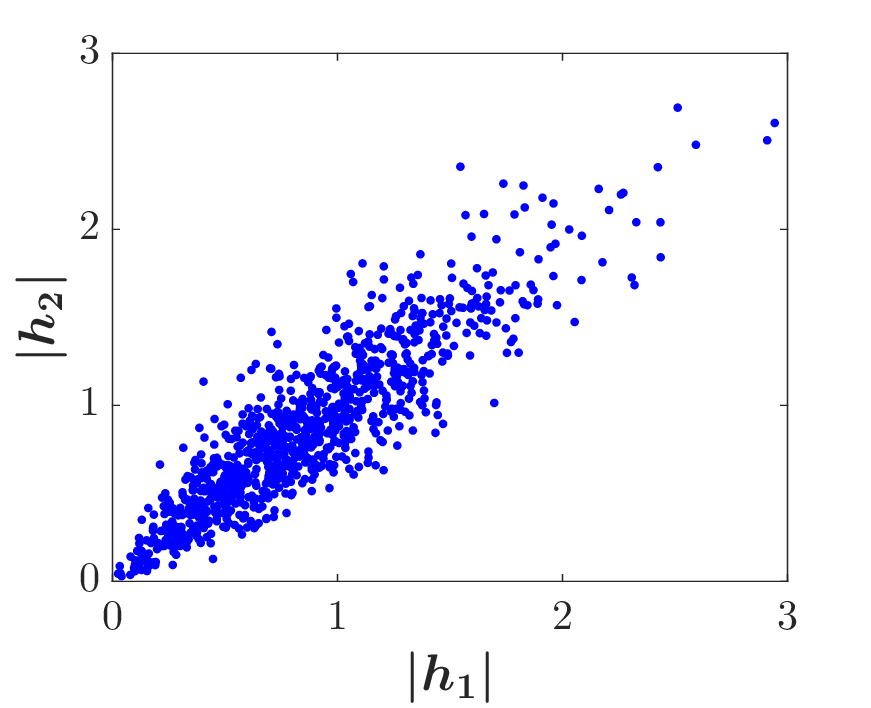}\label{sub-h2}%
}\hspace{-0.37cm}
\subfigure[$W=0.5$]{%
\includegraphics[width=0.26\textwidth]{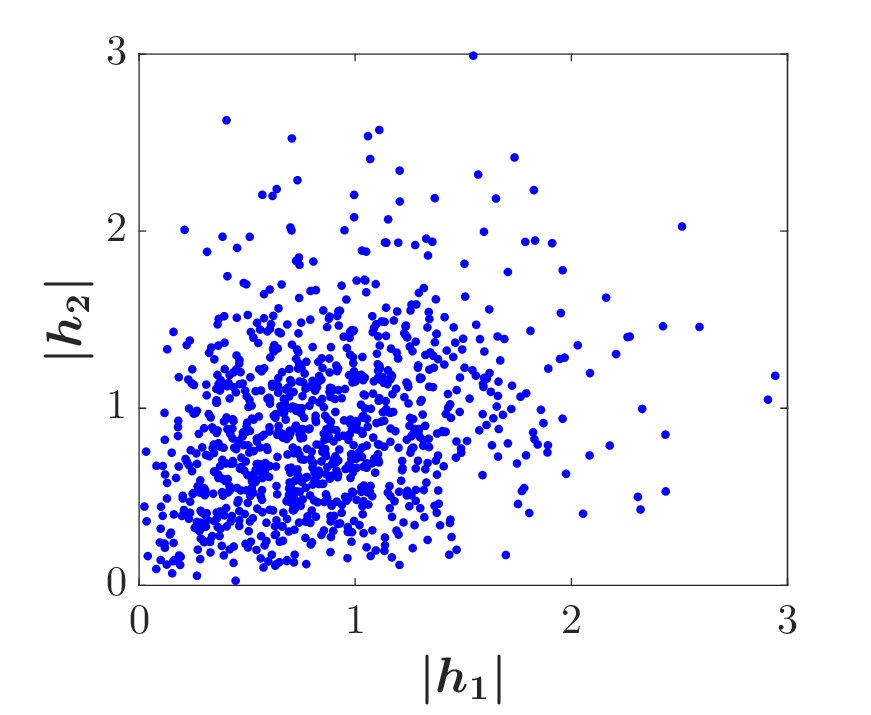}\label{sub-h3}%
}\hspace{-0.37cm}%or more
\subfigure[$W=4$]{%
\includegraphics[width=0.26\textwidth]{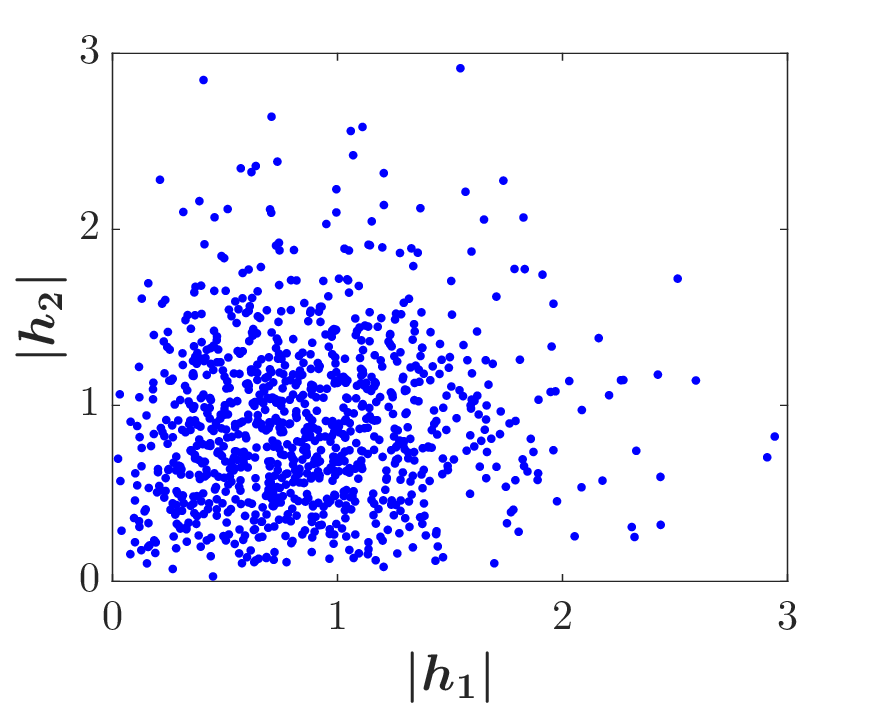}\label{sub-h4}%
}\\% 
\subfigure[$W=0.05$]{%
\includegraphics[width=0.26\textwidth]{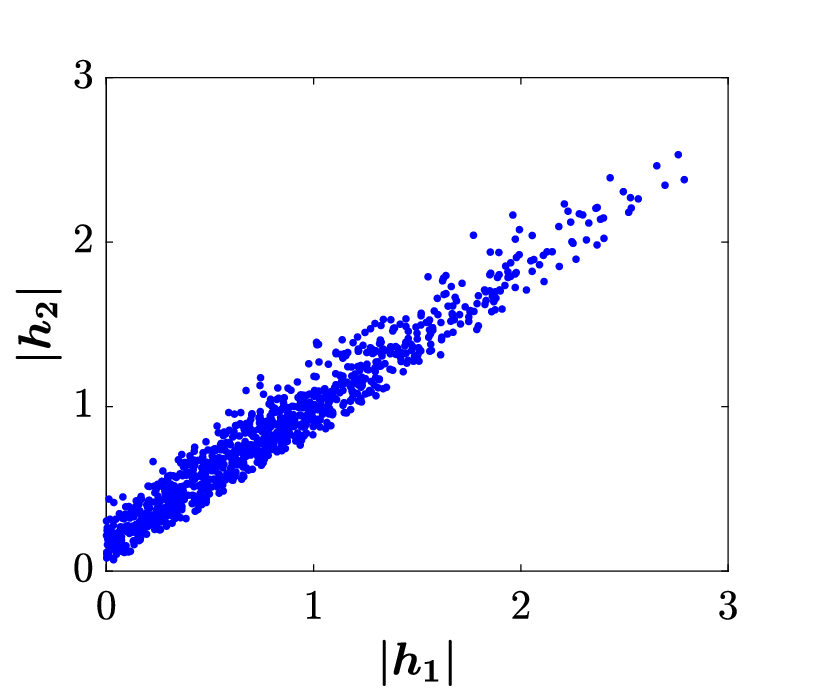}\label{sub-h5}%
}\hspace{-0.37cm}
\subfigure[$W=0.1$]{%
\includegraphics[width=0.26\textwidth]{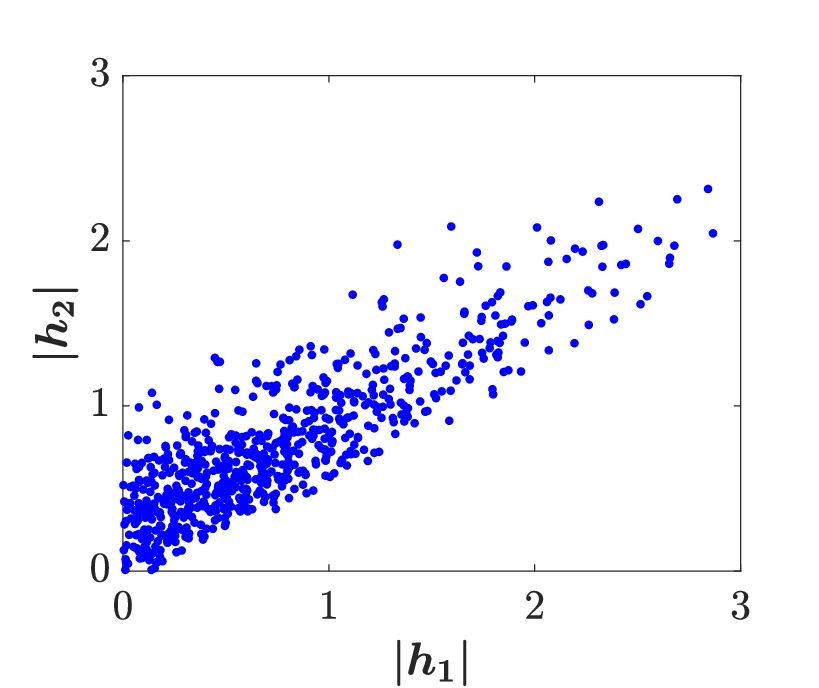}\label{sub-h6}%
}\hspace{-0.37cm}
\subfigure[$W=0.5$]{%
\includegraphics[width=0.26\textwidth]{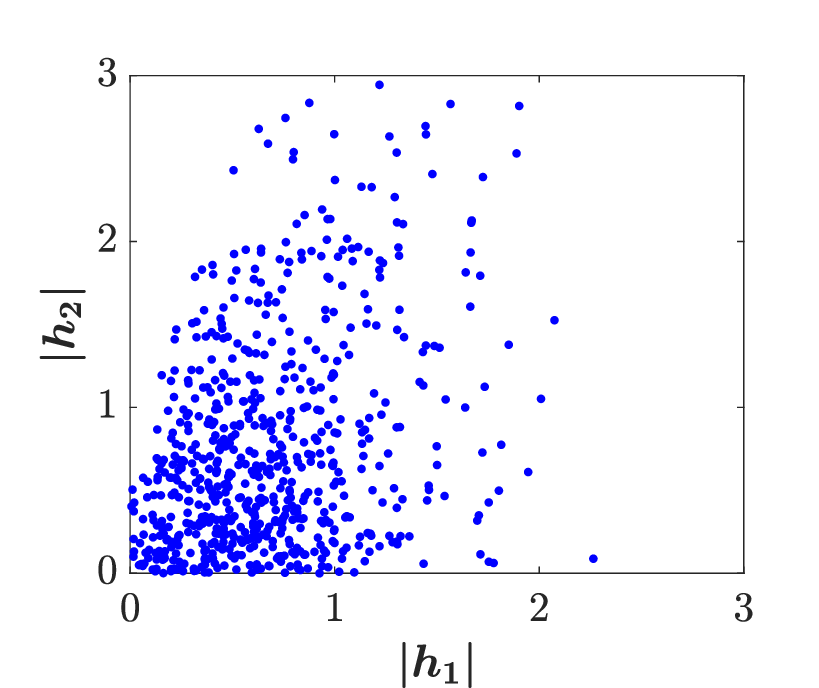}\label{sub-h7}%
}\hspace{-0.37cm}%or more
\subfigure[$W=4$]{%
\includegraphics[width=0.26\textwidth]{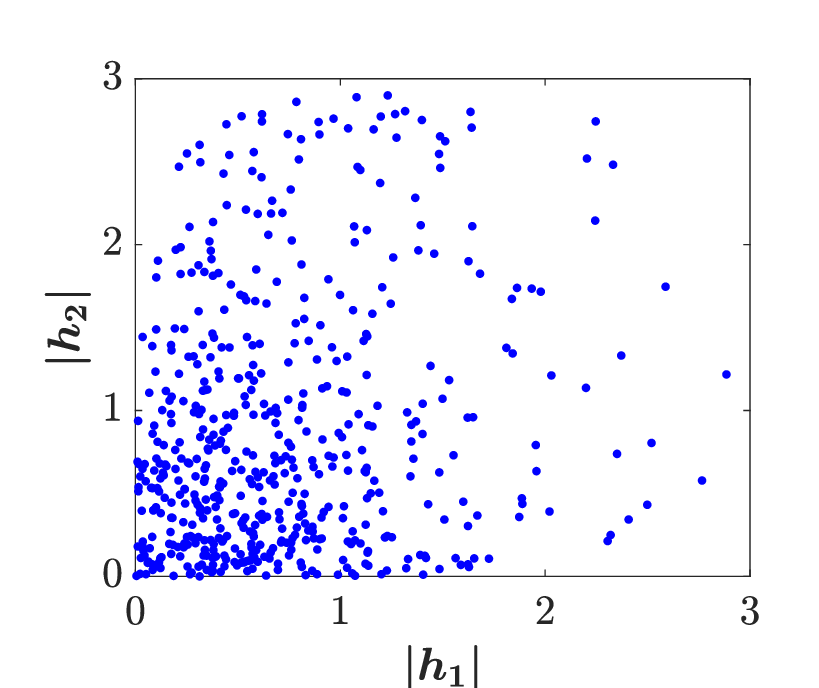}\label{sub-h8}%
}%  
\caption{Scatterplots describe the structure of dependency between two correlated Nakagami-$m$ fading channels  $|h_1|$ and $|h_2|$ when $m=1$ under: (a)--(d) Gaussian copula, and (e)--(h) Jakes' model.}\label{fig-scatter-h}\vspace{0cm}
\end{figure*}

\begin{remark}\label{remark4}
It is worth mentioning that the spatial correlation between correlated Nakagami-$m$ fading coefficients $|h_1|$ and $|h_2|$ under Gaussian copula can be measured by applying the inverse CDF of Nakagami-$m$ to the uniformly distributed RVs $u_1$ and $u_2$ (i.e., $|h_1|=F^{{\mathrm{Nak}^{-1}}}_{|h_1|}\left(u_1\right)$ and  $|h_2|=F^{{\mathrm{Nak}^{-1}}}_{|h_2|}\left(u_2\right)$). Following this approach, Fig.~\ref{fig-scatter-h} illustrates the scatterplots of fading coefficients $|h_1|$ and $|h_ 2|$ with Nakagami-$m$ fading under the Gaussian copula (i.e., Figs.~\ref{sub-h1}--\ref{sub-h4}) and Jakes' model (i.e., Figs.~\ref{sub-h5}--\ref{sub-h8}). A comparison between these two sets of scatterplots reveals that the Gaussian copula model closely mirrors the dependency structure observed in Jakes' model, confirming that for small (large) fluid antenna size $W$, the fading channel coefficients are highly (slightly) correlated. Moreover, it is evident that for various values of $W$, despite the Gaussian copula effectively capturing the overall trend and distribution shape observed in Jakes' model, there are slight differences in how tail dependences are captured. However, these differences will later prove not to have a major impact when evaluating system performance. Therefore, this comparison suggests that the Gaussian copula is capable of effectively modeling the dependency structure between the fading channels, providing a close approximation to Jakes' model across various scenarios.
\end{remark}

\begin{remark}\label{remark5}
By carefully observing the results in Figs.~\ref{subfig-cdf-k-w} and \ref{subfig-pdf-k-w}, it is evident that our Gaussian copula-based  model effectively mimics the CDF and PDF curves of Jakes' model, where it aligns well with Jakes' model, demonstrating a similar trend in fading behavior. However, slight differences exist, indicating that while the proposed model is designed to replicate Jakes' model, it introduces minor adjustments that can better fit specific cases. Despite these subtle variations, the overall performance of the proposed model suggests that it is a reliable approximation that preserves the main characteristics of the standard Jakes' model, such as the fading amplitude distribution. Moreover, we can see that as $K$ changes from $4$ to $8$ for small values of $W$, the CDF and PDF remain almost constant, which means that the FAS performance does not necessarily improve as $K$ grows. However, for large values of $W$, the CDF and PDF remarkably shift to the right as $K$ increases, i.e., the performance improves. In addition, we observe that as $W$ grows for a fixed value of $K$ (e.g., $K=8$), the CDF and PDF significantly shifts to the right, meaning that for constant $K$, increasing the fluid antenna size improves considerably the FAS performance.
\end{remark}

\begin{corollary}
The OP for the considered FAS under correlated Nakagami-$m$ fading channels by exploiting Gaussian copula is given by
\begin{multline}\label{eq-op-nak}
P_\mathrm{out}^\mathrm{Nak}=\Phi_{\vec{R}_{h_k,h_l}}\left(\sqrt{2}\,\mathrm{erf}^{-1}\left(\frac{2\gamma\left(m,\frac{m}{\mu}\hat{\gamma}^2\right)}{\Gamma(m)}-1\right)\right.\\
\left.,\dots,\sqrt{2}\,\mathrm{erf}^{-1}\left(\frac{2\gamma\left(m,\frac{m}{\mu}\hat{\gamma}^2\right)}{\Gamma(m)}-1\right)\right).
\end{multline}
\end{corollary}

\begin{proof}
By assuming $r=\hat{\gamma}$ in \eqref{eq-cdf-g-nak}, the proof is completed. 
\end{proof}

\begin{corollary}
The DOR for the considered FAS under correlated Nakagami-$m$ fading channels by exploiting Gaussian copula is given by
\begin{multline}\label{eq-dor-nak}
P_\mathrm{dor}^\mathrm{Nak}=\Phi_{\vec{R}_{h_k,h_l}}\left(\sqrt{2}\,\mathrm{erf}^{-1}\left(\frac{2\gamma\left(m,\frac{m}{\mu}\hat{T}^2\right)}{\Gamma(m)}-1\right)\right.\\
\left.,\dots,\sqrt{2}\,\mathrm{erf}^{-1}\left(\frac{2\gamma\left(m,\frac{m}{\mu}\hat{T}^2\right)}{\Gamma(m)}-1\right)\right).
\end{multline}
\end{corollary}

\begin{proof}
By assuming $r=\hat{T}$ in \eqref{eq-cdf-g-nak}, the proof is done. 
\end{proof}

\begin{figure}[t!]
\centering
\subfigure[]{%
\includegraphics[width=0.43\textwidth]{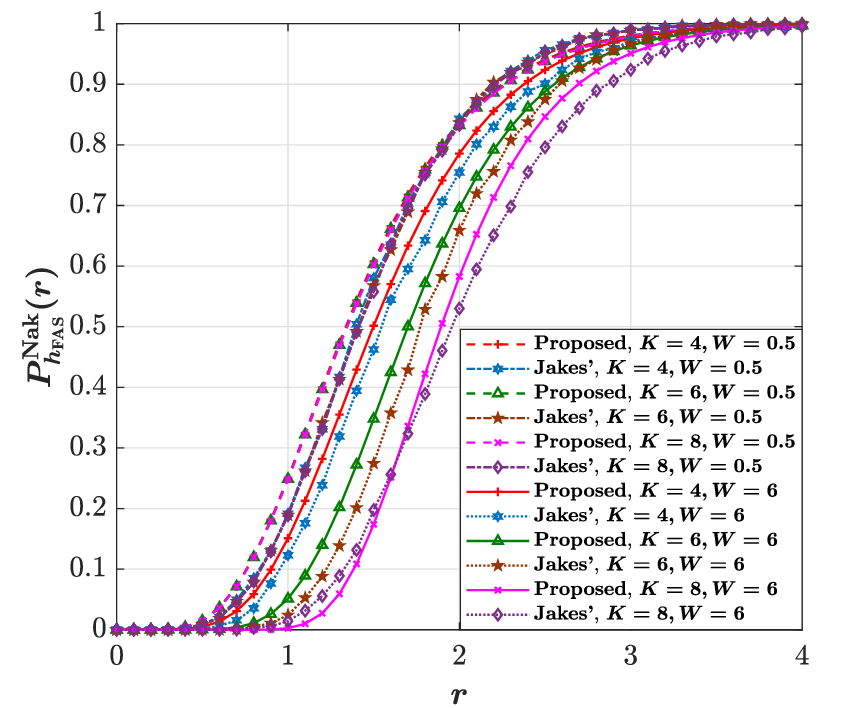}\label{subfig-cdf-k-w}%
}\hspace{-0.37cm}
\subfigure[]{%
\includegraphics[width=0.43\textwidth]{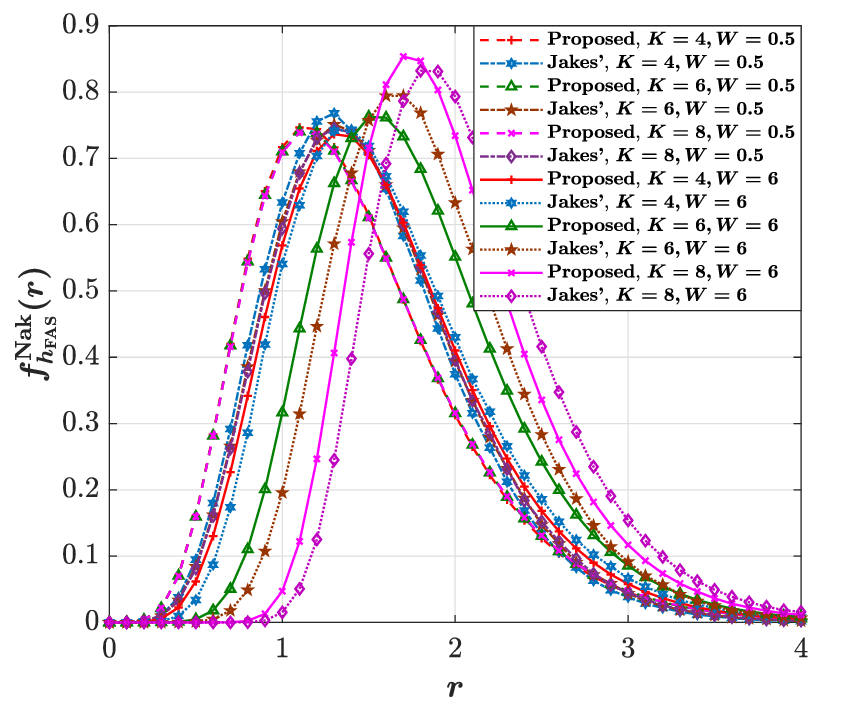}\label{subfig-pdf-k-w}%
}{\caption{(a) CDF and (b) PDF of FAS for selected values of $K$ and $W$ when $m=3$ and $\mu=1$.}}\label{fig-dist}\vspace{0cm}
\end{figure}

\section{Numerical Results}\label{sec-num}
In this section, we present numerical results to gain more insight into the OP and DOR performance of FAS. It should be noted that the analytical results obtained in \eqref{eq-op-nak} and \eqref{eq-dor-nak} are expressed in terms of the CDF of a multivariate normal distribution that technically has no closed-form expression. Nonetheless, it can be estimated numerically through various algorithms or implemented by the mathematical package of programming languages such as MATLAB, Python, and R. Additionally, as shown in Algorithm 1, the Gaussian copula can be simulated by applying the Cholesky decomposition of the given correlation matrix $\vec{R}$ to obtain the lower triangular matrix $\vec{A}$, such that $\vec{A}\vec{A}^T=\vec{R}$ \cite{mai2017simulating}. We also consider the conventional single-input single-output (SISO) fixed-antenna system as a benchmark to compare with the proposed FAS.

\begin{algorithm}
\caption{Gaussian Copula Simulation}
\textbf{Step 1.} \textit{Compute $\vec{A}$, such that $\vec{A}\vec{A}^T=\vec{R}$}\\
\textbf{Step 2.} \textit{Generate $\vec{s}=\left(S_1,\dots,S_d\right)$,  such that $S_i\sim\mathcal{N}\left(0,1\right)$ for $i=1,\dots,d$}\\
\textbf{Step 3.} \textit{Calculate $V_i=\sum_{j=1}^{i} A_{i,j}S_i$ for $i=1,\dots,d$}\\
\textbf{Step 4. }\textit{Return $U_i=\Phi_\vec{R}\left(V_i\right)$ for $i=1,\dots,d$}
\end{algorithm}
 
\begin{figure*}[t!]
\centering
\subfigure[]{%
\includegraphics[width=0.34\textwidth]{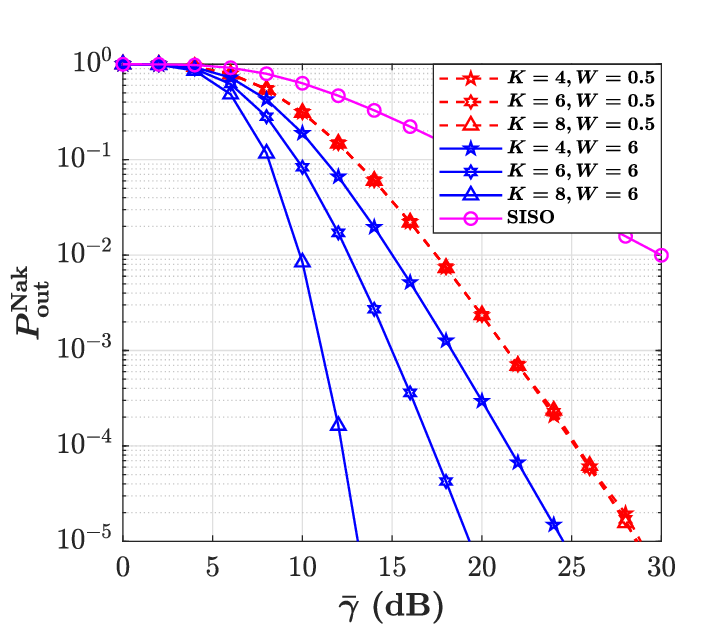}\label{subfig-out-k-w2}%
}\hspace{-0.37cm}
\subfigure[]{%
\includegraphics[width=0.34\textwidth]{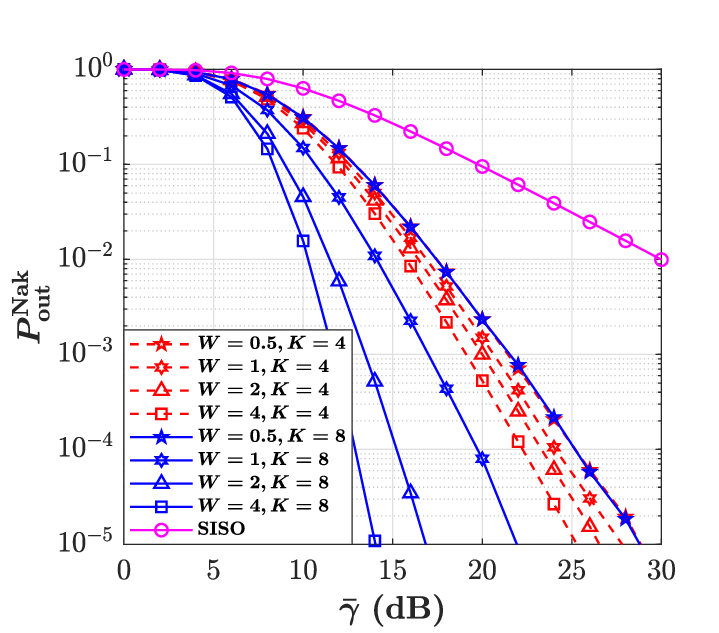}\label{subfig-out-k-w3}%
}\hspace{-0.37cm}
\subfigure[]{%
\includegraphics[width=0.34\textwidth]{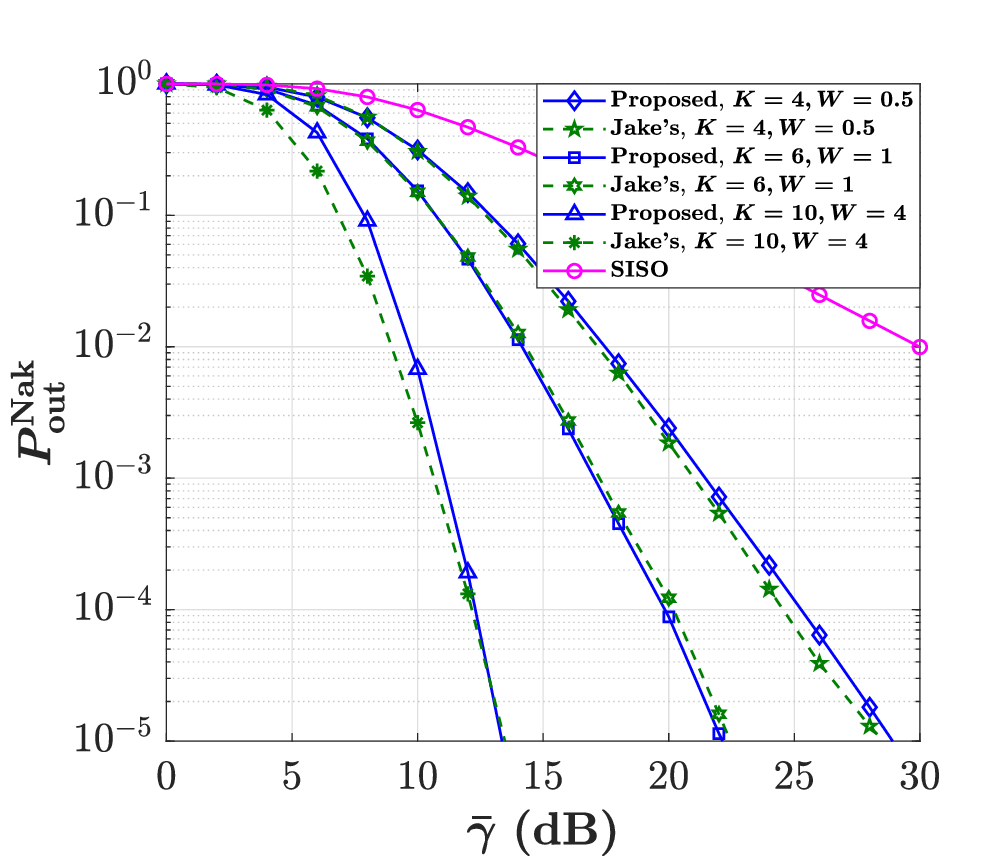}\label{subfig-out-k-w}%
}\hspace{-0.37cm}%or more
\caption{OP versus average transmit SNR $\bar{\gamma}$ for selected values of $W$ and $K$ when $\gamma_\mathrm{th}=10$dB, $m=1$, and $\mu=1$.}\label{fig-out-snr1}
\end{figure*}

\begin{figure*}[t!]
\centering
\subfigure[]{%
\includegraphics[width=0.34\textwidth]{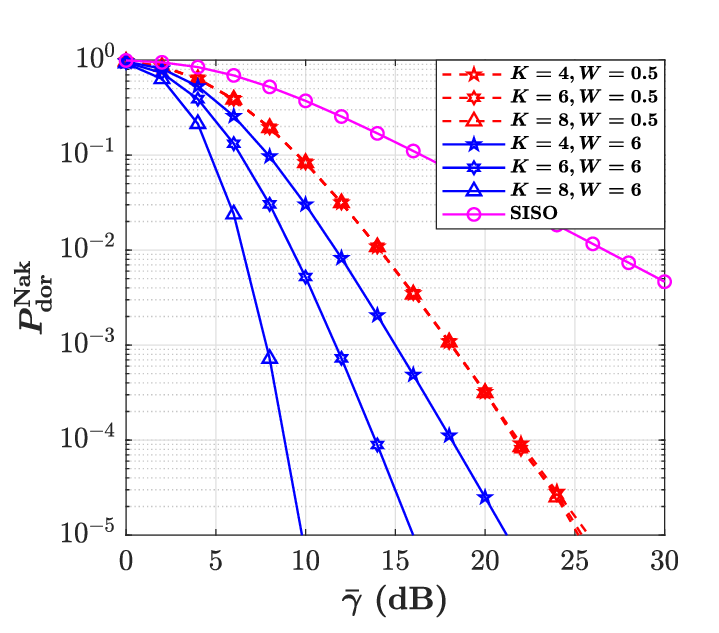}\label{subfig-dor-k-w2}%
}\hspace{-0.37cm}
\subfigure[]{%
\includegraphics[width=0.34\textwidth]{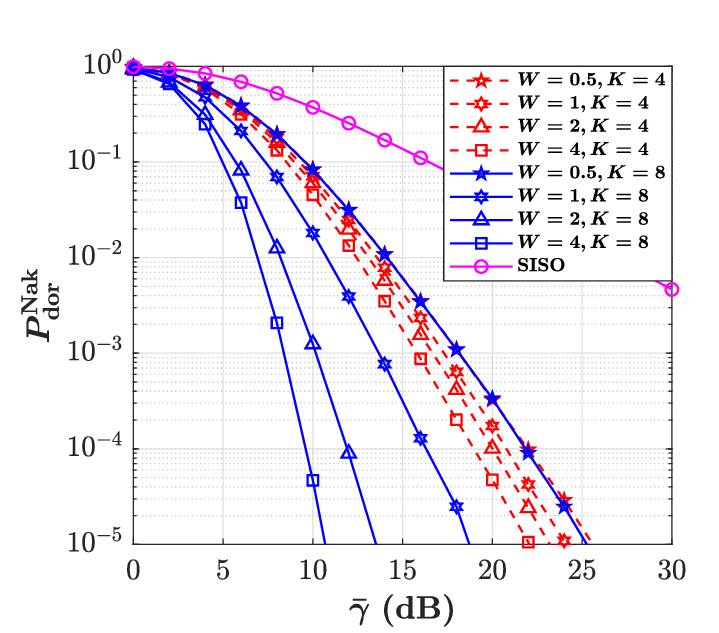}\label{subfig-dor-k-w3}%
}\hspace{-0.37cm}
\subfigure[]{%
\includegraphics[width=0.34\textwidth]{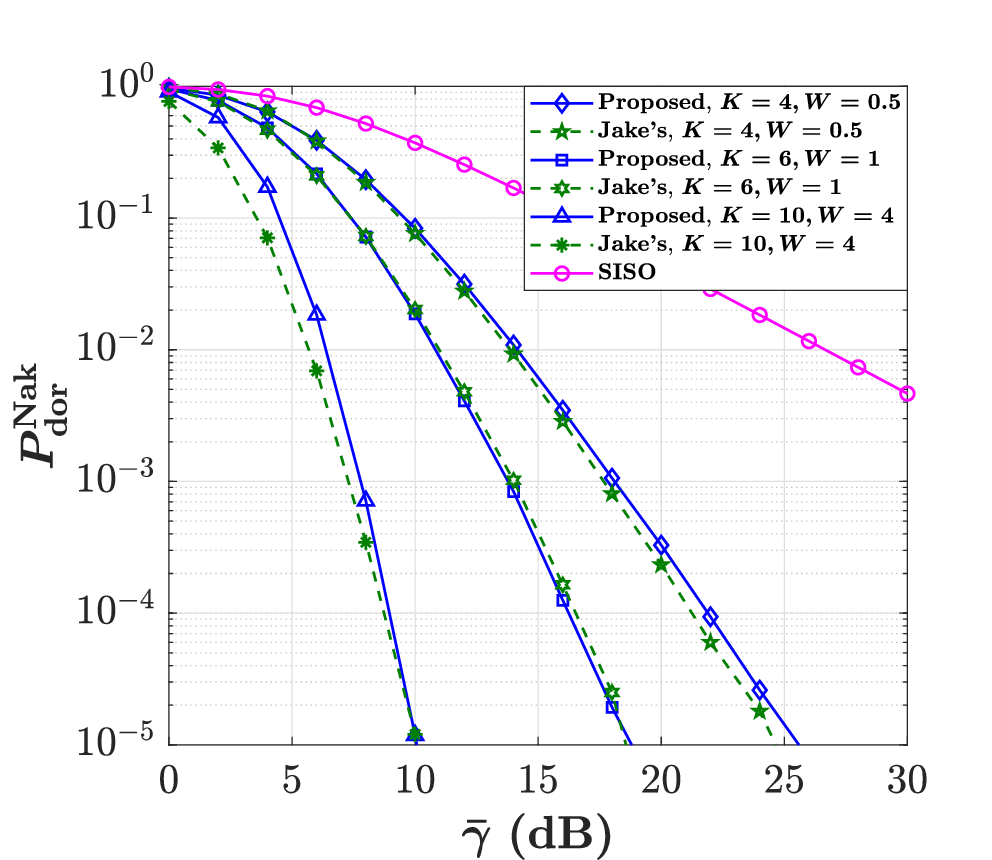}\label{subfig-dor-k-w}%
}\hspace{-0.37cm}%or more
\caption{DOR versus average transmit SNR $\bar{\gamma}$ for selected values of $W$ and $K$ when $\gamma_\mathrm{th}=10$dB, $B=2$MHz, $R=5$Kbits, $T_\mathrm{th}=3$ms, $m=1$, and $\mu=1$.\label{fig-k-w}}\vspace{0cm}
\end{figure*}

The behavior of the OP and DOR versus the average transmit SNR $\bar{\gamma}$ for given numbers of fluid antenna ports $K$ and selected values of fluid antenna size $W$ under correlated Nakagami-$m$ fading channels is illustrated in Figs.~\ref{fig-out-snr1} and \ref{fig-k-w}, respectively. As expected, we can see that the OP and DOR decrease as the average transmit SNR increases. From Figs.~\ref{subfig-out-k-w2} and \ref{subfig-dor-k-w2}, it can also be observed that as $K$ grows, the performance of OP and DOR improves for large values of fluid antenna size (e.g., $W=6$) which is in alignment with the findings of \cite{khammassi2023new}. However, increasing $K$ does not affect the system performance for small values of fluid antenna size (e.g., $W=0.5$), meaning that the OP and DOR remain almost constant. In other words, when $K$ rises for a fixed $W$, the space between ports decreases and the spatial correlation between them increases, and thus lower diversity gain is reached until eventually saturated. Furthermore, we can observe from Figs.~\ref{subfig-out-k-w3} and \ref{subfig-dor-k-w3} that increasing the fluid antenna size $W$ for larger $K$ provides more noticeable effects on the improvement of OP and DOR compared with smaller $K$. Moreover, it can be seen that increasing the spatial separation between the fluid antenna ports by increasing $W$ for a fixed $K$, and hence reducing the spatial correlation, can greatly ameliorate the system performance in terms of the OP and DOR. Figs.~\ref{subfig-out-k-w} and \ref{subfig-dor-k-w} indicate that by simultaneously increasing $W$ and $K$, spatial correlation between fluid antenna ports becomes balanced, and consequently lower OP and DOR are achieved. Besides, it can be seen that the FAS provides better performance in terms of the OP and DOR compared with the conventional SISO system in all scenarios even if the fluid antenna has large $K$ and small space. The main reason behind this improvement is due the capability of FAS in switching to the best port within a finite size $W$. Furthermore, in both Figs.~\ref{subfig-out-k-w} and \ref{subfig-dor-k-w}, we can observe that our Gaussian copula-based  model closely follows the behavior of Jakes' model across different configurations of $K$ and $W$. For each parameter set, the OP and DOR curves for the proposed model and Jakes' model overlap or are very close, demonstrating that our proposed model provides a highly accurate approximation (see Remarks \ref{remark4} and \ref{remark5}). While the curves for the Gaussian copula model and Jakes' model are mostly aligned, small deviations can be observed at certain SNR levels. These minor differences suggest that while the proposed model is highly accurate, it is still an approximation and might have slight estimation errors under certain conditions. 

The impact of fading parameter $m$ on the performance of OP and DOR for given values of $K$ and $W$ under correlated Nakagami-$m$ fading channels is investigated in Figs.~\ref{subfig-o-m}~and~\ref{subfig-d-m}, respectively. It is clearly seen that the efficiency of the OP and DOR improves under a mild fading condition (e.g., $m=3$) than when a stronger one (e.g., $m=0.5$) is considered, especially when the average SNR $\bar{\gamma}$ grows. Regarding the importance of fluid antenna size $W$ in realistic scenarios, we evaluate the behavior of the OP and DOR in terms of $W$ for different values of $\bar{\gamma}$ in Figs.~\ref{subfig-o-w-n} and \ref{subfig-o-w-n}, respectively. As expected, increasing the spatial separation between the fluid antenna ports for a fixed $K$ can provide lower values of the OP and DOR. We can also see that even under small values of $\bar{\gamma}$ and $W$ such as $0.5$, FAS offers better performance than the SISO system. 

\begin{figure*}[t!]
\centering
\subfigure[]{%
\includegraphics[width=0.45\textwidth]{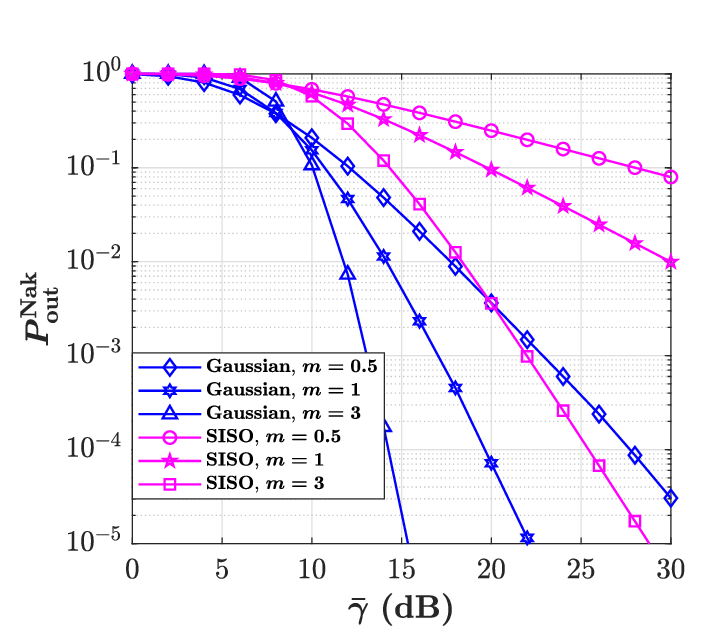}\label{subfig-o-m}%
}\hspace{0.9cm}
\subfigure[]{%
\includegraphics[width=0.45\textwidth]{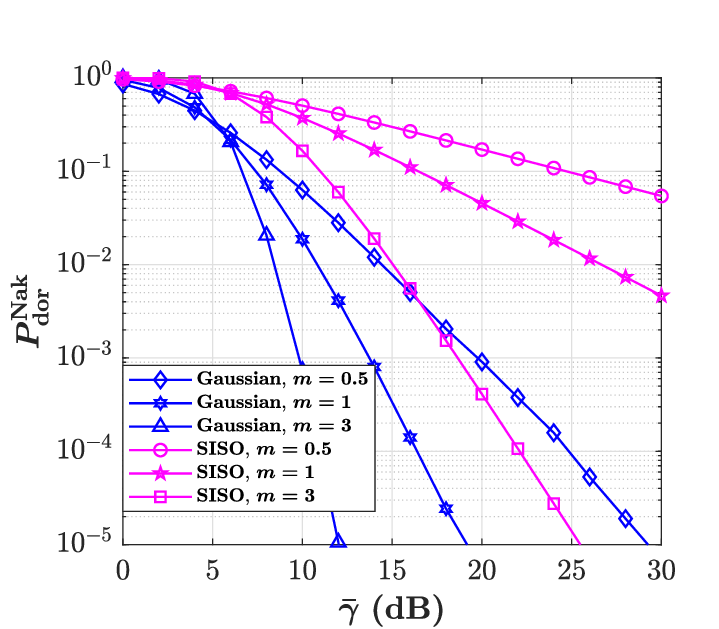}\label{subfig-d-m}%
}\hspace{0cm}%or more
\caption{(a) OP and (b) DOR versus average transmit SNR $\bar{\gamma}$ for selected values of fading parameter $m$ when $\gamma_\mathrm{th}=10$dB, $K=3$, $W=2.5$, $B=2$MHz, $R=5$Kbits, $T_\mathrm{th}=3$ms, and $\mu=1$.}\label{fig-o-d-m}\vspace{0cm}
\end{figure*}

\begin{figure*}[t!]
\centering
\subfigure[]{%
\includegraphics[width=0.45\textwidth]{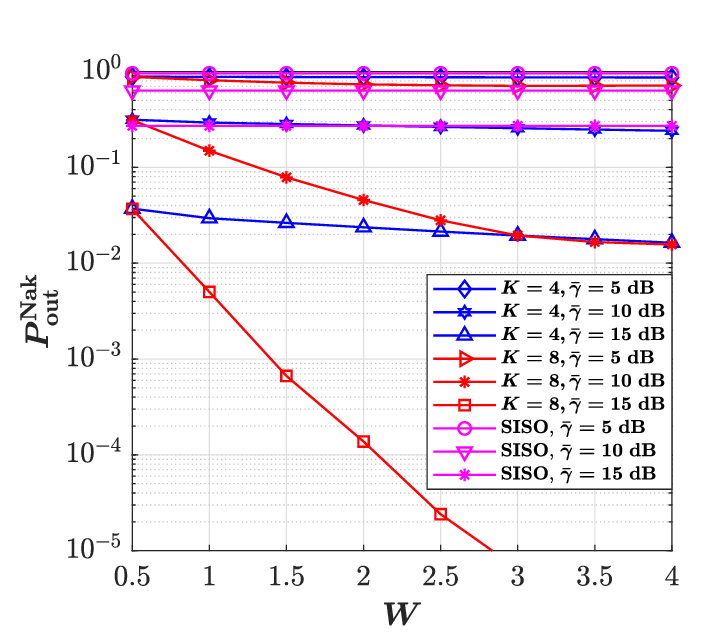}\label{subfig-o-w-n}%
}\hspace{0.9cm}
\subfigure[]{%
\includegraphics[width=0.45\textwidth]{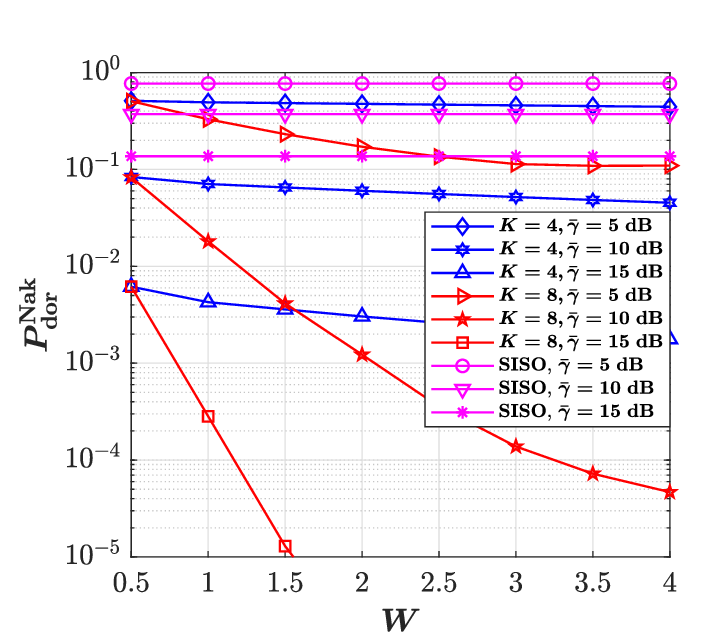}\label{subfig-d-w-k}%
}\hspace{0cm}%or more
\caption{(a) OP and (b) DOR versus fluid antenna size $W$  for selected values of fading parameter $K$ and $\bar{\gamma}$ when $\gamma_\mathrm{th}=10$dB, $B=2$MHz, $R=5$Kbits, $T_\mathrm{th}=3$ms, and $\mu=1$.}\label{fig-o-d-w}\vspace{0cm}
\end{figure*}

\begin{figure*}[t!]
\centering
\subfigure[]{%
\includegraphics[width=0.45\textwidth]{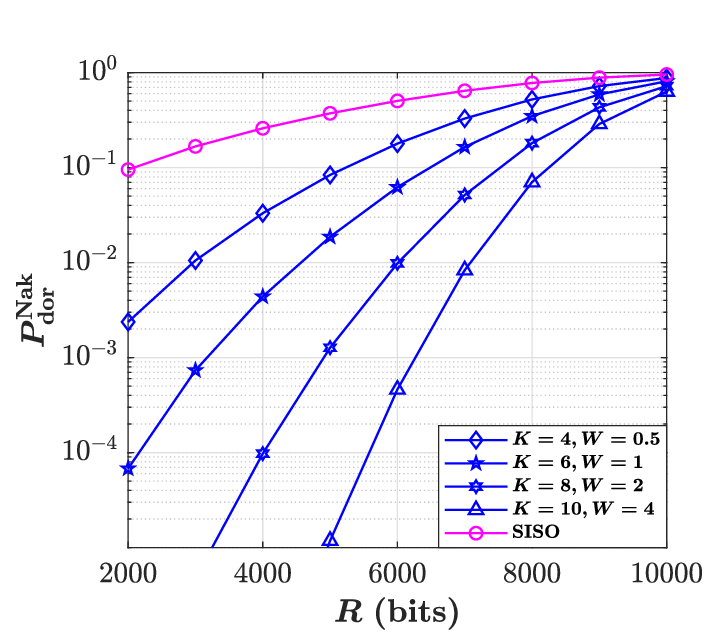}\label{subfig-dor-r-k-w}%
}\hspace{0.9cm}
\subfigure[]{%
\includegraphics[width=0.45\textwidth]{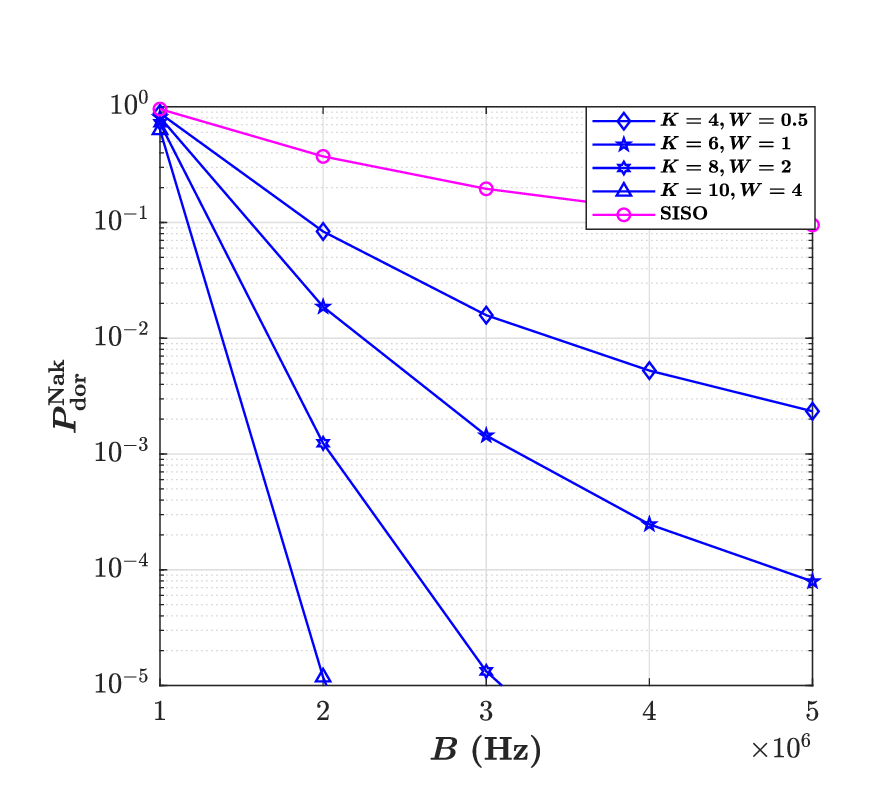}\label{subfig-dor-b-k-w}%
}\hspace{0cm}%or more
\caption{(a) DOR versus amount of data $R$ and (b) DOR versus bandwidth $B$ for selected values of fading parameter $K$ and $\bar{\gamma}$ when $\gamma_\mathrm{th}=10$dB, $B=2$MHz, $R=5$Kbits, $T_\mathrm{th}=3$ms, and $\mu=1$.}\label{fig-r-k-w}\vspace{0cm}
\end{figure*}

Fig.~\ref{subfig-dor-r-k-w} shows the DOR performance against the amount of transmitted data $R$ for different values of $W$ and $K$ in FAS. As expected, we can see that the DOR performance becomes worse as $R$ increases. Nonetheless, increasing the number of fluid antenna ports as well as the fluid antenna size can lead to a lower DOR when a fixed amount of data is sent. For instance, sending $R=6$Kbits amounts of data becomes almost impossible with low delay when $W$ and $K$ are small or when the SISO system is considered, but it can be transmitted with small delay if FAS is considered with large values of $W$ and $K$. The impact of bandwidth $B$ variations on the DOR performance for the considered FAS is illustrated in Fig.~\ref{subfig-dor-b-k-w}. It is obvious that the preset amounts of data can be sent with a lower delay when the channel bandwidth increases. Additionally, it can be seen that for a fixed $B$, the DOR performance significantly improves as $W$ and $K$ increase, meaning that data can be transmitted with a lower delay in FAS compared with the SISO system.  

\section{Conclusions}\label{sec-con}
In this paper, we studied the performance of FAS under arbitrary correlated fading channel coefficients, in which we exploited  copula theory to describe the dependence structure between the fluid antenna ports. To this end, we first expressed Jakes' model in terms of Gaussian copula and then derived compact analytical expressions for the CDF, PDF, OP, and DOR in terms of multivariate normal distribution for a general case, which are valid for any arbitrary choice of fading distribution. Next, we quantified the spatial correlation between the fluid antenna ports with the help of popular rank correlation coefficients (i.e., Spearman's $\rho$ and Kendall's $\tau$) and indicated the accuracy of using Gaussian copula in FAS. To analyze the performance in typical scenarios, we obtained analytical expressions of the CDF, PDF, OP, and DOR under correlated Nakagami-$m$ fading channels. Numerical results indicated that the system performance highly depends on the antenna size and the number of ports. Increasing the fluid antenna size provides lower OP and DOR but increasing the number of ports does not necessarily offer a better performance.   

%\bibliographystyle{IEEEtran}
%\bibliography{sample.bib}


\end{document}